\documentclass[a4paper,UKenglish,cleveref, autoref, thm-restate]{lipics-v2021}
\pdfoutput=1 %uncomment to ensure pdflatex processing (mandatatory e.g. to submit to arXiv)
\hideLIPIcs  %uncomment to remove references to LIPIcs series (logo, DOI, ...), e.g. when preparing a pre-final version to be uploaded to arXiv or another public repository

\usepackage{amsmath}
\usepackage{cite}

\usepackage{todonotes}

\usepackage{amsthm}
\usepackage{amssymb}
\usepackage{bbm}
\usepackage{amsfonts}
\usepackage{mathrsfs}
\usepackage{wrapfig} %wrap the text around the picture

\usepackage{cleveref}
\usepackage{nicefrac}
\usepackage{todonotes}
\usepackage{amsfonts}
\usepackage[capitalise]{cleveref} 
\newcommand{\Prob}{\mathbb{P}}
\newcommand{\Exp}{\mathbb{E}}
\newcommand{\Geom}{\ensuremath{\text{Geo}}}
\newcommand{\Bin}{\ensuremath{\text{Bin}}}
\newcommand{\Var}{\text{Var}}
\newcommand{\tmx}[1]{\ensuremath{\tau^{\text{max}}_{#1}}}
\newcommand{\tmn}[1]{\ensuremath{\tau^{\text{min}}_{#1}}}
\newcommand{\LabEx}{A}
\newcommand{\ContDel}[1]{\ensuremath{B^{(#1)}}}
\newcommand{\ContFull}[1]{\ensuremath{C^{(#1)}}}
\newcommand{\ContSum}[1]{\ensuremath{\mathbf{C}_{#1}}}
\newcommand{\DiscDel}[2]{\ensuremath{D^{(#1),#2}}}
\newcommand{\DiscFull}[2]{\ensuremath{F^{(#1),#2}}}
\newcommand{\DiscSum}[2]{\ensuremath{\mathbf{F}_{(#1),#2}}}
\newcommand{\ball}[3]{\mathcal{B}_{#1}\left(#2,#3\right)}
\newcommand{\distribution}{label multiplicity distribution}

\usepackage{thmtools} 
\usepackage{thm-restate}

\usepackage{caption} %for subfigure - join multiple figures and add captions
\usepackage{subcaption}
\usepackage{changepage} %inside figure we can put \begin{adjustwidth}{-1cm}{-1cm} \end{adjustwidth} and we can ignore the margins
\usepackage{mathtools} % write text [under]{over} arrow $\xrightarrow[\text{world}]{\text{hello}}$

\usepackage{graphicx}
\graphicspath{{Figures/}}

\usepackage{etoolbox}
\newcommand{\appsymb}{$\star$}
\newcommand{\appref}[1]{{\hyperref[proof:#1]{\appsymb}}}
\newcommand{\appLink}[1]{{\hyperref[#1]{\appsymb}}}

\title{Sharp Thresholds for Temporal Motifs and Doubling Time in Random Temporal Graphs} %TODO Please add

\titlerunning{Sharp Thresholds: Temporal Motifs and Doubling Time in Random Temporal Graphs} %TODO optional, please use if title is longer than one line

\author{Henry Austin}{Department of Computer Science, Durham University, UK}{henry.b.austin@durham.ac.uk}{https://orcid.org/0009-0007-5685-4944}{}

\author{George B. Mertzios}{Department of Computer Science, Durham University, UK}{george.mertzios@durham.ac.uk}{https://orcid.org/0000-0001-7182-585X}{Supported by the EPSRC grant EP/P020372/1.}

\author{Paul G. Spirakis}{Department of Computer Science, University of Liverpool, UK}{p.spirakis@liverpool.ac.uk}{https://orcid.org/0000-0001-5396-3749}{Supported by the EPSRC
grant EP/P02002X/1.}

\authorrunning{Henry Austin, George B. Mertzios, and Paul G. Spirakis} %TODO mandatory. First: Use abbreviated first/middle names. Second (only in severe cases): Use first author plus 'et al.'

\Copyright{Henry Austin, George B. Mertzios, and Paul G. Spirakis} %TODO mandatory, please use full first names. LIPIcs license is "CC-BY";  http://creativecommons.org/licenses/by/3.0/

 %TODO mandatory; please add comma-separated list of keywords

\category{} %optional, e.g. invited paper

\relatedversion{} %optional, e.g. full version hosted on arXiv, HAL, or other respository/website
\nolinenumbers %uncomment to disable line numbering
\hideLIPIcs

\EventEditors{John Q. Open and Joan R. Access}
\EventNoEds{2}
\EventLongTitle{42nd Conference on Very Important Topics (CVIT 2016)}
\EventShortTitle{CVIT 2016}
\EventAcronym{CVIT}
\EventYear{2016}
\EventDate{December 24--27, 2016}
\EventLocation{Little Whinging, United Kingdom}
\EventLogo{}
\SeriesVolume{42}
\ArticleNo{23}
\begin{document}
\maketitle

\begin{abstract}
In this paper we study two natural models of \textit{random temporal} graphs.
In the first, the \textit{continuous} model, each edge $e$ is assigned~$l_e$ labels, each drawn uniformly at random from $(0,1]$, where the numbers $l_e$ are independent random variables following the same discrete probability distribution. 
In the second, the \textit{discrete} model, 
the $l_e$ labels of each edge $e$ are chosen uniformly at random from a set $\{1,2,\ldots,T\}$. 
In both models we study the existence of \textit{$\delta$-temporal motifs}.
Here a $\delta$-temporal motif consists of a pair $(H,P)$, where $H$ is a fixed static graph and $P$ is a partial order over its edges. A temporal graph $\mathcal{G}=(G,\lambda)$ contains $(H,P)$ as a $\delta$-temporal motif if $\mathcal{G}$ has a simple temporal subgraph on the edges of $H$ whose time labels are ordered according to $P$, and whose life duration is at most $\delta$. 
We prove \textit{sharp existence thresholds} for all $\delta$-temporal motifs, and we identify a qualitatively different behavior from the analogous static thresholds in Erd\H{o}s-R\'{e}nyi random graphs. 
Applying the same techniques, we then characterize the growth of the largest $\delta$-temporal clique in the continuous variant of our random temporal graphs model. 
Finally, we consider the \textit{doubling~time} of the reachability ball centered on a small set of vertices of the random temporal graph as a natural proxy for temporal expansion. 
We prove \textit{sharp upper and lower bounds} for the maximum doubling time in the continuous model.\\

\end{abstract}

\section{Introduction}

% As the natural model for expressing binary relationships between objects, it is difficult to overstate the ubiquity of graphs throughout the study of algorithms, combinatorics and the modeling of real world systems.
% However, when we consider dynamic objects, where such relationships may change over time, the natural model is the much less well-studied ``temporal graph''.
% A temporal graph may be viewed as a static graph augmented with a \emph{labeling function} mapping each edge to a set of \emph{labels} each corresponding to a time when the edge is present.

% While temporal graphs have been studied in a variety of contexts ranging from sociology, to epidemiology, to even evolutionary game theory, the primary focus of most theoretical work concerns cataloging differences between the properties of static and temporal graphs. This program of research has been highly successful, uncovering both structural and algorithmic distinctions from the static case. 

A temporal (or dynamic) graph is a graph whose topology is subject to discrete changes over time. 
This paradigm reflects the structure and operation of a vast variety of modern real-life networks; 
social networks, wired or wireless networks whose links change dynamically, transportation networks, and several physical systems are only a few examples of networks that change over time 
\cite{michailCACM,Nicosia-book-chapter-13,holme2019temporal,KMS-MFCS23, Santoro11}. 
Embarking from the fundamental model for temporal graphs, introduced by Kempe et al.~\cite{kempe00} with one (integer) time-label per edge, and from its extension which allows multiple (integer) time-edges per edge~\cite{MertziosMCS13}, we consider here a relaxed version of the model where the time-labels are drawn from the positive real numbers in~$\mathbb{R}_{+}$ (see \cref{temp-graph-def}). 
This relaxation of the model allows us to randomly draw each time label uniformly at random from a time interval, e.g.~from $(0,1]$.

\begin{definition}[temporal graph]
\label{temp-graph-def} A \emph{temporal graph} is a pair $\mathcal{G}=(G,\lambda)$,
where $G=(V,E)$ is a footprint (static) graph with $n$ denoting the number of vertices in $V$, and $\lambda :E\rightarrow 2^{\mathbb{R}_{+}}$ is a \emph{labeling} function that assigns to every edge of $G$ a finite set of time labels, each being a positive real number. 
$\mathcal{G}$ is a \emph{simple} temporal graph if every edge has exactly one time label, i.e.~the labeling function is $\lambda :E\rightarrow \mathbb{R}_{+}$.
\end{definition}

Temporal graphs are traditionally used to model information flow which, due to causality, can only pass along a sequence of edges with increasing time labels. 
Motivated by this, the majority of work concerning temporal graphs is based 
on the notion of temporal paths and other
``path-related'' notions, such as temporal analogues of distance, diameter,
reachability, exploration, and centrality~\cite{AkridaGMS16,ErlebachHK15,MertziosMCS13,MichailS-TSP16,AkridaGMS-TOCS17,EnrightMMZ21,MolterRZ24,EnrightMM25,DogeasEKMM24,RymarMNN23,KlobasMMS25,KlobasMMNZ23}. 
Only relatively recently have attempts been made to define ``non-path'' temporal graph problems, such as temporal analogues of clique \cite{viardClique, viardCliqueTCS,neidermeier,DavotEMM25}, 
vertex cover \cite{AkridaMSZ18,HammKMS22}, 
coloring \cite{MertziosMZ21,MarinoS22,AdamsonMS26,IbiapinaNRR25}, 
matching \cite{MertziosMNZZ20,BasteBR20,CioniD0S025}, 
vertex separators \cite{MaackMNR23,FluschnikMNRZ20,HarutyunyanKP23}, 
and transitivity \cite{MertziosMRSZ25}. 

Continuing in this line of research on ``non-path'' temporal problems, we mainly focus in this paper on the existence of a given (arbitrary) \emph{temporal motif} in a temporal graph. 
A temporal motif consists of a pair $(H,P)$, where $H$ is a fixed (static) graph and $P$ is a partial ordering over the edges of $H$. 
A temporal graph $\mathcal{G}=(G,\lambda)$ is said to contain a temporal motif $(H,P)$, if $H$ is contained in the footprint $G$, and the edges of $H$ have time labels which are ordered in agreement with the partial order $P$. 
The existence and frequency of temporal motifs is a question of both theoretical and practical importance. On the one hand, theoretically, temporal motifs stand as an intermediate object containing less temporal information than a temporal subgraph, but more temporal information than a footprint subgraph. In particular, they are well suited for capturing notions of causality and are less sensitive to the precise numerical timings of labels. On the other hand, temporal motifs are widely used in the practical analysis of temporal networks, for both identifying significant interaction patterns and determining the mechanics of network evolution \cite{sariyuce2025powerful}. An especially common technique in this context is comparing the frequency of temporal motifs in a given empirical network with that of some appropriate (and frequently ad hoc) network\cite{sariyuce2025powerful, doi:10.1137/19M1242252, kovanen2011temporal}. 
A strengthening of this concept, which we study, is the notion of a \emph{$\delta$-temporal motif} (i.e.~one where all of its constituent edge labels must lie within some window of length~$\delta$). These, in turn, represent a middle ground between a temporal subgraph and a temporal motif. They have seen application both in fraud detection~\cite{FraudDetection} and as a tool for studying different interaction scales in empirical networks~\cite{DeltaTemporalMotifs}.

Additionally, the vast majority of theoretical work has focused on deterministic worst-case temporal graphs.
While this has led to a relatively good understanding of the worst-case temporal graph properties, the approach provides little information about the vast majority of temporal graphs.
This distinction was recently and dramatically demonstrated, by the discovery that almost all simple temporal graphs permit nearly optimal spanners \cite{SharpThresholds}, despite the existence of temporal graphs without sub-quadratic spanners \cite{axiotis_et_al:LIPIcs.ICALP.2016.149}. 
Results such as these raise the general question: to what extent do such unique properties hold for \textit{random temporal graphs} rather than only for carefully constructed ones?

\subsection{Our Contribution}

In this paper we focus on non-path problems on randomly generated temporal graphs. We expand on earlier work by \cite{Cliques,SharpThresholds}, by exploring the properties of two model variants (one continuous and one discrete model) for random temporal graphs, generalizing the Random Simple Temporal Graph Model (RSTG) \cite{SharpThresholds,Cliques}.
For ease of reading we defer the formal definition of the models (see \cref{continuous-def,discrete-def}).
Intuitively, in the continuous model we sample the time labels uniformly at random from $(0,1]$ for each edge, where the number of labels on each edge is sampled independently for each edge from some discrete distribution.
The discrete model is the natural discretization of the continuous one, where the labels of each edge are chosen uniformly at random from a set $\{1,2,\ldots,T\}$. 
In our first result, we explore the existence thresholds for all fixed $\delta$-temporal motifs, in terms of the length $\delta$ of the permitted time window.
Since the temporal graphs we consider here typically have dense foot-print graphs, but disparate time labels, this serves as a useful, if imperfect, analogue for the existence thresholds found in sparse static random graphs.
In fact we identify a behavior that is similar, but subtly different from the existence thresholds of fixed static graphs in the Erd\H{o}s-R\'{e}nyi random graph model. 
In particular, the threshold is determined \textit{not} by the reciprocal of the density of the densest subgraph, but instead by the following quantity which we call the \emph{sparsity}.

\begin{definition}
    For a graph $G=(V,E)$, the \emph{sparsity} $\rho_G$ of $G$ is defined as follows,
    \[\rho_G=\min_{H \sqsubseteq G}\frac{|V_H|}{|E_H|-1},\]
    where $H \sqsubseteq G$ denotes that $H$ is a subgraph of $G$ containing at least $2$ edges and $V_H$ (resp. $E_H$) are the set of vertices (resp. edges) comprising $H$.
\end{definition}

Intuitively, this stems from the fact that a random temporal graph has many different intervals which could potentially contain a given $\delta$-temporal motif.
Individually, these intervals contain motifs at thresholds roughly in accordance with the existence of their footprint in an Erd\H{o}s-R\'{e}nyi random graph and the number of such intervals needed to cover the lifetime of a given temporal graph is inversely proportional to the length of the intervals.
However, a direct proof along these lines is made challenging due to the dependency between these intervals, especially when the density of the graph is itself a random variable.

This seemingly minor change, from the reciprocal of density to sparsity, is sufficient to cause qualitatively different threshold behavior from the sparse static case. For example, $\delta$-temporal motifs built from cycles of different sizes have asymptotically distinct existence thresholds in our model, unlike their static counterparts \cite{frieze2015introduction} in the Erd\H{o}s-R\'{e}nyi random graph.
Another consequence of this result is that the asymptotic thresholds are independent of the edge ordering of a given motif. 
This implies, for example, that time-respecting paths of a given length have the same asymptotic threshold as non-time-respecting paths of the same length.
Furthermore, we extend this behavior beyond fixed motifs to consider the size of the largest temporal motif with a clique footprint, and obtain bounds generalizing the results of \cite{Cliques}.

Our second main result considers a different celebrated property of static random graphs: their higher order connectivity. 
Since the existence of a good analogue for expansion in temporal graphs remains an open problem, we use the maximum doubling time of the reachability ball centered on small sets of vertices as a natural proxy. 
In our continuous model, when the number of labels per edge is fixed, we provide sharp upper and lower bounds on the maximum doubling time. 
As a direct consequence, we immediately obtain a sharp threshold for finite maximum doubling time in the Random Simple Temporal Graph (RSTG) model of \cite{SharpThresholds}. 
Interestingly, this corresponds precisely to the threshold for the existence of temporal source and sink vertices in that model. In contrast, the doubling time of any specific large set of vertices corresponds to the threshold for both point-to-point reachability and the existence of a giant component in the RSTG.
\section{Related Work}

Most research concerning temporal motifs has focused on either algorithmic techniques for counting or estimating the number of temporal motifs, or the application of these methodologies for the analysis of real networks \cite{sariyuce2025powerful}. A frequent issue in the latter is the construction of a ``null'' or ``reference'' model to compare the relative frequency of motifs against \cite{sariyuce2025powerful, kovanen2011temporal, doi:10.1137/19M1242252}. Despite this there has been very limited research into the distribution of motifs in random temporal graphs. A notable exception to this is \cite{AnalyticalModels}, which derives algorithms for calculating the expectation and the variance of the number of motifs appearing in a proposed random temporal graph model. While closely related, their work differs from ours in several ways. Firstly, they focus primarily on the construction of a model that can be statistically fit to data, and so use a stochastic block model footprint graph with Poisson distributed inter-edge appearance times. This permits them a greater degree of topological richness than is obtained in our simpler model, however the distribution of the labels on each edge is more restricted. Secondly, they focus primarily on providing an algorithm for the calculation of the first and second moment, rather than their application to find existence thresholds. Thirdly, they work only in the continuous setting and do not consider anything comparable to our discrete model. An interesting open problem is closing this gap, in particular providing equivalent results in a unified model capturing both the topological expressivity of the stochastic block model and temporal expressivity of our \distribution{} based approach. While not presented as a result on temporal motifs, \cite{Cliques} investigates the growth of the largest $\delta$-temporal clique, which implicitly provides thresholds for motifs with a clique footprint and a trivial partial order. In this work, we strictly generalize these results to the more powerful random graph models we consider.

Until recently, as discussed in \cite{SharpThresholds}, reachability problems in random temporal graphs have primarily been studied in other fields, such as population protocols, the gossip model or the random edge ordering model. For example, in population protocols the reachability properties of the underlying temporal graph underpin two of the model's most fundamental lower bounds: the time complexity of broadcast \cite{DBLP:conf/nca/MocquardSRA16} and leader election \cite{DBLP:journals/corr/abs-1906-11121}. However, recently, in part due to the seminal work of \cite{SharpThresholds}, there has been growing interest in the reachability properties of random temporal graphs from both a structural and algorithmic perspective. While other models have been considered (see for example \cite{baguley2026temporal,https://doi.org/10.1002/rsa.70040, brandenberger2025temporal}), there has been a focus on establishing the properties of the so-called ``Random Simple Temporal Graph'' defined as an Erd\H{o}s-R\'{e}nyi random graph with a linear order over its edges. Sharp thresholds have since been established at: $\frac{\log{n}}{n}$ for point-to-point reachability \cite{SharpThresholds} and the existence of a giant \cite{ReviewerSpecialty} (or even super constant \cite{atamanchuk2024size}) connected-component; $\frac{2\log{n}}{n}$ for the existence of temporal sources and sinks; $\frac{3\log{n}}{n}$ for full temporal connectivity; and $\frac{4\log{n}}{n}$ for the existence of pivotal (and non-sharply optimal) spanners \cite{SharpThresholds}. Additionally, there have been several results concerning various characteristics of the length of time-respecting paths in the RSTG and their connections to connectivity \cite{angel2020long, broutin2024increasing}. Our treatment of doubling time shares the most similarity with \cite{ReviewerSpecialty} as, similarly to there, we require bounds on the growth of reachability forests, however rather than using a foremost forest (as in \cite{ReviewerSpecialty, SharpThresholds}) or the random recursive tree (as in \cite{broutin2024increasing}), we instead derive the times directly via a simple coupling with a sum of geometric random variables. An additional distinction is our choice of models which are not necessarily simple and have asymptotically constant density. Despite this, when restricted to the interval $[a,a+p]$ with a single label per edge, sampling from our continuous model is equivalent to sampling from the RSTG with parameter $p$.
\section{Model and Definitions}

As a natural generalization of an important tool for the analysis of static graphs, some notion of a ``temporal motif'' has been studied in a variety of contexts and with a variety of definitions (see \cite{sariyuce2025powerful} for a survey). Particular points of distinction relevant to this work include whether an occurrence of a temporal motif must represent an induced subgraph (see \cite{DeltaTemporalMotifs} and \cite{AnalyticalModels} for a difference in opinion) and whether the associated order must be partial \cite{kovanen2011temporal}, or total \cite{DeltaTemporalMotifs, AnalyticalModels,SamplingMethodsForCountingTemporalMotifs, 10.1093/bioinformatics/btv227}. We make use of the following definition for (simple) $\delta$-temporal motifs.

\begin{definition}[$\delta$-temporal motif]
\label{motif-def}
    Let $H=(V_H,E_H)$ be a finite static graph, $\delta \in \mathbb{R}$ be a number and $P=(E_H,\prec)$ be a partial order over the edges of $H$. Then a temporal graph $\mathcal{G}=(V,E,\lambda)$ contains $(H,P)$ as a \emph{$\delta$-temporal motif} if and only if there exists a temporal subgraph $\mathcal{I}=(V',E',\lambda')$ and a mapping $\Phi: V_H\to V'$ such that:
    \begin{itemize}
        \item $\mathcal{I}$ \textbf{is a simple temporal subgraph of} $\mathcal{G}$: $V'\subseteq V$, $E'\subseteq E[V']$ and $\forall e \in E': \lambda'(e) \in \lambda(e)$
        \item $\Phi$ \textbf{is an isomorphism of} $H$ \textbf{and realizes the order over the labels}: $\Phi$ is an isomorphism of $H$, $\forall uv,wx \in E_H: uv \prec_{P} wx \implies \lambda'(\Phi(u)\Phi(v))< \lambda'(\Phi(w)\Phi(x))$. \footnote{\cref{motif-def} is given under the assumption that $P$ is a strict partial order. Alternatively a $\delta$-temporal motif could be defined with respect to a \textit{non-strict} partial order $P$, that is, $uv \prec_{P} wx \implies \lambda'(\Phi(u)\Phi(v))\leq \lambda'(\Phi(w)\Phi(x))$. With this alternative definition, we believe, but have not verified, that all results of the paper should remain the same, excluding the case where two labels are required to be equal by the order.}
        \item $\mathcal{I}$ \textbf{has a life duration of at most} $\delta$: $\max_{e_1,e_2 \in E'}\lambda'(e_1)-\lambda'(e_2)<\delta$.
    \end{itemize}
\end{definition}

To avoid trivialities we restrict ourselves to considering motifs containing at least $2$ edges, and describe any graph containing $1$ or fewer edges as \emph{trivial} \footnote{We note that the empty graph is a $\delta$-temporal motif of any sufficiently large graph. Furthermore, all simple motifs containing only a single edge appear as $\delta$-temporal motifs with thresholds $0$ and $1$ in the continuous and discrete models respectively.}. 
For two graphs, $G$ and $H$ we use $G \subseteq H$ to indicate that $G$ is a (non-empty) subgraph of $H$ and $G\sqsubseteq H$ to indicate that $G$ is a non-trivial subgraph of $H$.

In this paper we investigate two models of a random temporal graph. The first is a continuous model, while the second is a natural discretization of it. In both models, there are two forms of randomness: the \textit{number of labels} on each edge (which is sampled identically and independently at random from the ``\distribution{}'') and the \textit{labels themselves} (which are sampled uniformly at random from some interval depending on the model).
\begin{definition}
    We define the \emph{label multiplicity distribution} %%\distribution{} 
    $\psi$ to be any distribution over the non-negative integers with a positive and finite second moment.
    %such that $\psi$ is non-zero with positive probability and both the first and second moments exist.
\end{definition}
For a random variable $Q\sim \psi$, we denote $r=\Exp[Q]$ and $r_2=\Exp[Q^2]$. Furthermore, we denote $q=\Prob[Q>0]$.
We can now define both our continuous and discrete models.
\begin{definition}[Doubly Random Temporal Graph; continuous case]
\label{continuous-def}
    We define the distribution $\Gamma_{V}(\psi)$ of random temporal graphs on the vertex set $V$ with \distribution{} $\psi$. A sample $\mathcal{G}=(V,E,\lambda)$ from $\Gamma_V(\psi)$ is obtained as follows: 
    For each edge $e \in \binom{V}{2}$, we sample the set of edge occurrences $\lambda(e)$ by drawing $l_e \sim \psi$ samples uniformly at random from %%%the interval 
    $(0,1]$.
\end{definition}

\begin{definition}[Doubly Random Temporal Graph; discrete case]
\label{discrete-def}
    We define the distribution $\Gamma_{V}(\psi,T)$ of random temporal graphs on the vertex set $V$ with \distribution{}~$\psi$ and largest possible time-label $T$. A sample $\mathcal{G}=(V,E,\lambda)$ from $\Gamma_V(\psi,T)$ is obtained as follows:
    For each edge $e \in \binom{V}{2}$, we sample the set of edge occurrences $\lambda(e)$ by drawing $l_e \sim \psi$ samples uniformly at random from the set $\{1,...,T\}$.
\end{definition}
Throughout the paper,  
in the discrete model, we will only consider $\beta$-temporal $(H,P)$ motifs, where $\beta$ is at least $\text{height}(P)$, the length of the longest chain in $P$. 
The reason for this is that, if $\beta$ is smaller than the longest chain in $P$, then trivially a $\beta$-temporal motif $(H,P)$ does not exist. 
There is a natural relationship between the two models, based on the following notion.

\begin{definition}
    For a continuous temporal graph $\mathcal{G}=(V,E,\lambda)$ with all labels restricted to $(0,1]$, its \emph{$T$-discretization} is given by $\mathcal{H}=(V,E,\lambda')$ where $\lambda'(e)=\{\lceil T l\rceil : l \in \lambda(e)\}$.
\end{definition}
We immediately have the following observation connecting the two models.
\begin{observation}
    \label{obs: bridge 1}
    For all vertex sets $V$ and \distribution s $\psi$ the distribution $\Gamma_V(\psi,T)$ is identical to the distribution of $T$-discretizations of a random sample from~$\Gamma_V(\psi)$.
\end{observation}

Furthermore, we consider the higher order connectivity of these temporal graphs. 
Due to the lack of a good definition for the notion of expansion for temporal graphs, we choose to use the \textit{doubling time} of reachability balls as a proxy for expansion.
\begin{definition}
    \label{def: reach}
        Let $\mathcal{G}=(V,E,\lambda)$ be a temporal graph and $S\subset V$ be a set of vertices such that $|S|\leq \frac{|V|}{2}$. We define the following notions:
        \begin{enumerate}
            \item The $t$-step reachability ball $\ball{\mathcal{G}}{S}{t}$ of $S$ on $\mathcal{G}$ to be all vertices reachable via a time-respecting path initiated from a vertex in $S$ and arriving by time $t$,
            \item The doubling time $\text{Double}_{\mathcal{G}}(S)=\min\{t \in \mathbb{N}_0:|\ball{\mathcal{G}}{S}{t}|\geq 2|S|\}$ of a set of vertices to be the minimum time required for the reachability ball to double in size. If $\{t \in \mathbb{N}_0:|\ball{\mathcal{G}}{S}{t}|\geq 2|S|\}=\emptyset$, we define $\text{Double}_{\mathcal{G}}(S)=\infty$,
            \item The doubling time $\textbf{Double}(\mathcal{G})=\max_{S \subset V: |S|\leq \frac{|V|}{2}} \text{Double}_{\mathcal{G}}(S)$ of a temporal graph to be the maximum doubling time over all small sets of vertices.
        \end{enumerate}
    \end{definition}
Finally, we say that an event concerning a temporal graph occurs "with high probability" if the probability of that event approaches $1$ as the size of the temporal graph grows. 
\section{Results}
Our first result is on existence thresholds for fixed $\delta$-temporal motifs in the continuous variant of our model.
\begin{theorem}
    \label{thm: Continuous Existence}
    Let $H=(V_H,E_H)$ be a non-trivial static graph, $P=(E_H,\prec)$ be any partial order over the edges of $H$ and $\rho_H=\min_{I \sqsubseteq H}\frac{V_I}{E_I-1}$. A random temporal graph drawn from $\Gamma_{[n]}(\psi)$:
    \begin{itemize}
        \item Contains $(H,P)$ as a $\delta(n)$-temporal motif with high probability if $\delta(n)=\omega(n^{-\rho_H})$.
        \item Does not contain $(H,P)$ as a $\delta(n)$-temporal motif with high probability if $\delta(n)=o(n^{-\rho_H})$.
    \end{itemize}
\end{theorem}

For the discrete model, recall that, for any choice of a fixed graph $H$ and a partial order~$P$ over its edges, a sample from $\Gamma_{[n]}(\psi,T)$ cannot contain $(H,P)$ as a $\beta$-temporal motif, if~$\beta < \text{height}(P)$. 
Therefore we assume below without loss of generality that $\beta \geq \text{height}(P)$. 
We obtain the following equivalent result for the discrete model.
\begin{theorem}
    \label{thm: Discrete Existence}
    Let $H=(V_H,E_H)$ be a non-trivial static graph, $P=(E_H,\prec)$ be any partial order over the edges of $H$ and $\rho_H=\min_{I \sqsubseteq H}\frac{V_I}{E_I-1}$. Suppose that there exists $n_0>0$ such that $\min_{n\geq n_0}\beta(n)\geq \text{height}(P)$. Then, for $T(n)\geq\beta(n)$, a random temporal graph drawn from $\Gamma_{[n]}(\psi,T(n))$
    \begin{itemize}
        \item Contains $(H,P)$ as a $\beta(n)$-temporal motif with high probability if $\frac{\beta(n)}{T(n)}=\omega(n^{-\rho_H})$.
        \item Does not contain $(H,P)$ as a $\beta(n)$-temporal motif with high probability if $\frac{\beta(n)}{T(n)}=o(n^{-\rho_H})$.
    \end{itemize}
\end{theorem}

% We would like to emphasize the distinction between $\rho_H$ and the $\min_{I \subseteq H}\frac{V_I}{E_I}$ exponent seen in the roughly analogous sparse static graph results~\cite{frieze2015introduction}. 
% Although this is a seemingly small change in the form of $\rho_{H}$, it leads to a \textit{qualitatively different behavior} of the threshold that we obtain.
% For example, in the continuous model, the next result follows immediately from \cref{thm: Continuous Existence}.

% \begin{corollary}
%     Let $C_i$ be the cycle on $i$ vertices and $P_i$ be any partial order over its edge set. A random temporal graph drawn from $\Gamma_{[n]}(\psi)$:
%     \begin{itemize}
%         \item Contains $(C_i,P_i)$ as a $\delta(n)$-temporal motif with high probability if $\delta(n)=\omega(n^{-\frac{i}{i-1}})$.
%         \item Does not contain $(C_i,P_i)$ as a $\delta(n)$-temporal motif with high probability if $\delta(n)=o(n^{-\frac{i}{i-1}})$. 
%     \end{itemize}
% \end{corollary}

% Thus each finite cycle appears as a motif at a unique asymptotic threshold in the size of the window. However, every finite cycle appears at the same asymptotic threshold in the density of an Erd\H{o}s-R\'{e}nyi random graph \cite{frieze2015introduction}.
We also apply our techniques to generalize the previous work of \cite{Cliques} and investigate the growth of the largest $\delta$-temporal clique in the continuous variant of our model. We find that the introduction of more labels produces a subtly different behavior, although still close to the clique number of Erd\H{o}s-R\'{e}nyi random graphs. 
\begin{theorem}
    \label{thm: Cliques}
    Let $K_i$ be the clique on $i$ vertices and $P^{\emptyset}_i$ be the partial order over its edges such that they are all mutually incomparable. Further, let $0<\delta<1$ be a constant, $\mathcal{G}=(V,E,\lambda)$ be a sample from $\Gamma_{[n]}(\psi)$ and define $w_\delta(\mathcal{G})$ to be the maximum value such that $\mathcal{G}$ contains $(H_{w_{\delta}(\mathcal{G})},P^{\emptyset}_{w_\delta(\mathcal{G})})$ as a $\delta$-temporal motif. Then for any $\epsilon>0$, the following holds with high probability,
    \[(1-\epsilon)\frac{2\log{n}}{\log{\left(\frac{\delta r_2-\delta r +r}{\delta r^2}\right)}}\leq w_{\delta}(\mathcal{G})\leq(1+\epsilon)\frac{2\log{n}}{\log{\left(\frac{\delta r+(1-\delta)q}{\delta q r}\right)}} .\]
\end{theorem}
Finally, we consider the doubling time of a random sample from $\Gamma_{[n]}(\psi)$ in the special case where $\psi$ is a degenerate distribution, i.e.~where every edge has the same number of labels. 
% We prove the following.
\begin{theorem}
    \label{thm: continuous double}
    Let $\mathcal{G}=(V,E,\lambda)$ be a sample from $\Gamma_{[n]}(\psi)$, with $\psi$ a degenerate distribution. For any $\alpha>0$,
    \[\Prob\left[\frac{(2-\alpha)\log{n}}{rn}\leq \mathbf{Double}(\mathcal{G})\leq \frac{(2+\alpha)\log{n}}{rn}\right]=1-o(1).\]
\end{theorem}
We note that when $\psi$ is the degenerate distribution on $1$, this statement is equivalent to the following statement about the RSTG model,
\begin{corollary}
    Let $\mathcal{H}=(V,E,\lambda)$ be a sample from the RSTG distribution with parameter $\beta(n)$. Then the doubling time of $\mathcal{H}$ is asymptotically almost surely:
    \begin{itemize}
    \item finite, if $\beta(n)>(2+\Omega(1))\frac{\log{n}}{n}$,
    \item infinite, if $\beta(n)<(2-\Omega(1))\frac{\log{n}}{n}$.
    \end{itemize}
\end{corollary}

\section{Distribution of Temporal Motifs}
In this section we outline the proofs of our results (see \cref{thm: Continuous Existence,thm: Discrete Existence,thm: Cliques}) on the thresholds for the existence of motifs in temporal graphs. The results are obtained via a modification of a first and second moment argument, loosely following the treatment of their static analogues in \cite{frieze2015introduction}. Since the proofs of the results concerning both growing motifs and those of constant size share a lot of technical details (simply optimizing for different factors in the associated bounds), we present their shared components together.

\subsection{Technical Preliminaries}
Before we present our proofs, we will need a few preliminary results and definitions.
Our first preliminary is a short result concerning the probability of a given collection of labels obeying the partial order. Here it is important to distinguish between the continuous case, where labels are totally ordered almost surely, and the discrete case, where multiple labels may share the same time stamp.
\begin{restatable}{definition}{DefRPosetFunction}
    For a poset $P=(A,\prec)$ and $\beta \in [n]$, we take $R(P)$ to be the number of linear extensions of $P$ and $R(P,\beta)$ to be the number of mappings from $A$ to $[\beta]$ such that $P$ is satisfied on $[\beta]$.
\end{restatable}
\begin{restatable}{lemma}{LemRPosetExtension}
    \label{lem: Poset-Extension}
    For a finite non-empty set $A$ and a partial order $P=(A,\prec)$,
    \[\inf_{\beta\geq \text{height(P)}} \frac{R(P,\beta)}{\beta^{|A|}}\geq|A|^{-|A|}\]
    Similarly,
    \[\frac{R(P)}{|A|!}\geq |A|^{-|A|}\]
\end{restatable}
\begin{proof}
    The result for $R(P)$ holds trivially as $R(P)\geq 1$ and $|A|!\leq |A|^{|A|}$. 
    For the other inequality,
    we will separate into two cases, and show that the minimum in each is at least constant.
    Since $\beta>\text{height}(P)$, $R(P,\beta)>1$ (as the assignment of each element to its height trivially satisfies $P$).
    Thus, if $|A|>\beta$, \[\inf_{\text{height(P)}\leq\beta<|A|} \frac{R(P,\beta)}{\beta^{|A|}}\geq |A|^{-|A|}.\]

    Secondly, we must consider the case where $|A|\leq \beta$. 
    In order to obtain a bound, we will consider the strict total order $P'$ produced by some arbitrary linear extension of $P$.
    It follows that $R(P,\beta)>R(P',\beta)$ for all $\beta$, as any assignment satisfying $P'$ must satisfy $P$.
    However, $R(P',\beta)=\binom{\beta}{|A|}$ and thus for $|A|\leq \beta$, 
    \[\inf_{\beta\geq |A|} \frac{R(P,\beta)}{\beta^{|A|}}\geq \frac{\binom{\beta}{|A|}}{\beta^{|A|}}\geq |A|^{-|A|}.\]
    Thus, the inequality holds in either case.
\end{proof}
We will also need the following notion of automorphism for temporal motifs. In particular, we require that the mapping over the edges produced by the automorphism of the vertices corresponds to an automorphism of the associated partial order. The number of such automorphisms will play a crucial role in preventing the over-counting of the number of occurrences of temporal motifs when multiple occurrences share the same labels.
% \begin{restatable}{definition}{DefAutomorphism}
%     For a fixed graph $H=(V,E)$ and a partial order $P=(E,\prec)$, $\phi:V\to V$ is an automorphism of $(H,P)$ if $\phi$ is an automorphism of $H$ and $\Phi:\binom{V}{2}\to \binom{V}{2}$, such that $\Phi(\{a,b\})=\{\phi(a),\phi(b)\}$, is an automorphism of $P$. We denote the number of such automorphisms by $\#\text{Aut}(H,P)$.
% \end{restatable}
\begin{restatable}{definition}{DefAutomorphism}
    For a fixed graph $H=(V,E)$ and a partial order $P=(E,\prec)$, $\phi:V\to V$ is an automorphism of $(H,P)$ if $\phi$ is an automorphism of $H$ preserving the ordering of $P$.
    Explicitly, we require that (i) for all $u,v$ we have $\{u,v\} \in E_H\iff \{\phi(u),\phi(v)\} \in E_H$ and (ii) for all pairs of edges $\{u,v\}, \{w, x\} \in \binom{V}{2}$ we have $\{u,v\} \prec \{w,x\} \iff \{\phi(u)
    ,\phi(v)\} \prec \{\phi(w),\phi(x)\}$. We denote the number of such automorphisms by $\#\text{Aut}(H,P)$.
\end{restatable}
Additionally, we require the following simple domination result, which relates the probability of a temporal motif appearing in two instances of our random graph models based on stochastic domination between their \distribution s. This will serve as an extremely useful tool which allows us to reduce the problem of obtaining certain bounds for arbitrary \distribution s to the case of Bernoulli or otherwise bounded \distribution s.
\begin{restatable}{lemma}{LemDomination}
    \label{lem: domination}
    Let $I=(V_I,E_I)$ be a fixed graph, $P=(E_I,\prec)$ a partial order over its edges and $\psi$ and $\chi$ two valid \distribution s, such that $\psi$ is stochastically dominated by $\chi$. Further let $\mathcal{G}=(V,E,\lambda)$ be a temporal graph sampled from $\Gamma_V(\psi)$ and $\mathcal{H}=(V,E,\lambda')$ be a temporal graph sampled from $\Gamma_V(\chi)$. Then,
    \begin{itemize}
        \item For any $0<\delta<1$, the probability $(I,P)$ appears as a $\delta$-temporal motif in $\mathcal{G}$ is at most the probability that $(I,P)$ appears as a $\delta$-temporal motif in $\mathcal{H}$.
        \item For any $0<\beta<T$, the probability $(I,P)$ appears as a $\beta$-temporal motif in the $T$ discretization of $\mathcal{G}$ is at most the probability that $(I,P)$ appears as a $\beta$-temporal motif in the $T$ discretization of $\mathcal{H}$. \medskip
    \end{itemize}
\end{restatable}
\begin{proof}
    We present only the proof for the continuous case, as the discrete case is identical up to a change in notation, and begin with following simple observation.
    \begin{observation}
        Let $\mathcal{J}=(V,E,\lambda'')$ and $\mathcal{K}=(V,E,\lambda''')$ such that for all $e \in E$, $\lambda'''(e)\subseteq \lambda''(e)$. Any occurrence of $(I,P)$ as a $\delta$-temporal motif in $\mathcal{K}$ must also exist in $\mathcal{J}$.
    \end{observation}
    We will construct a coupling between $\mathcal{G}$ and $\mathcal{H}$ to ensure such monotonicity exists between them. Since $\chi$ stochastically dominates $\psi$, there exists a monotone coupling between them (i.e. one such that for $X\sim \psi$ and $Y\sim \chi$, $\Prob{[X\leq Y]}=1$). We can therefore couple $\mathcal{G}$ and $\mathcal{H}$ such that for all $e \in E$, $|\lambda(e)|< |\lambda'(e)|$ by adopting this monotone coupling for the distribution of the number of labels on each edge individually. Then by sampling the same set of label times for the labels on each edge, we obtain a coupling such that $\lambda(e)\subseteq \lambda'(e)$ for all $e \in E$.
    Taking $A$ to be the indicator for the event that $(I,P)$ exists as a $\delta$-temporal motif in $\mathcal{G}$ and $B$ the same for $\mathcal{H}$, we find that under our coupling,
    \[\Prob{[A\leq B]}=1.\]
    Thus, we have a stochastic order between the distribution of $A$ and $B$ and our claim.
\end{proof}
We will now construct the random variables of interest.
For any $S\subseteq (V,\binom{V}{2})$ we define $\mathcal{T}_S$ to be the set of all mappings from $E_S\to \mathbb{N}_0$.
Furthermore, for $\tau_S \in \mathcal{T}_S$ and $\tau_T \in \mathcal{T}_T$, where $S \subseteq T$, we define the order $\tau_S \leq \tau_T$ if and only if $\forall e \in E_{S}$ we have that $\tau_S(e)\leq \tau_T(e)$.
% Furthermore, we define the following order over $\cup_{I \subseteq G}\mathcal{T}_I$ where for $a \in \mathcal{T}_S$ and $b \in \mathcal{T}_T$ we describe $a\geq b$ if and only if $\forall e \in E_{S\cap T}$ $a(e)\geq b(e)$ and $S \cap T\neq \emptyset$.
We define $\tmx{S}$ to be the random function such that $\forall e \in E_S$ we have $\tmx{S}(e)=l_e$ and $\tmn{S}$ to be the constant mapping $\tmn{S}(e)=1$. 

We further define $\mathcal{S}$ to be the set of all occurrences of $H$ in the footprint graph $G=(V,E)$.
For mathematical convenience we shall treat $\lambda(e)$ as a tuple ordered by the order of sampling rather than a multi-set and we shall assume that for each edge the value of all labels are sampled first, only later followed by how many will actually appear.
This allows us to interpret any pair $(S,\tau)$ for $S \in \mathcal{S}$ and $\tau \in \mathcal{T}_S$, as a specific candidate for a motif, namely $S$ where each edge $e \in E_S$ has its $\tau(e)$th sampled label.
We may now define the following useful random variables:
\begin{definition}
    For a random sample $\mathcal{G}$ from $\Gamma_{[n]}(\psi)$, $S \in \mathcal{S}$ and $\tau \in \mathcal{T}_S$ we define:
    \begin{itemize}
        \item $\LabEx_{S,\tau}$ to be the indicator for the event $\{\forall e \in E_S:l_{e}\geq \tau(e)\}$ (i.e. the event that $\tau$ describes a set of labels that actually exist for $S$)
        \item $\ContDel{\delta}_{S,\tau}$ to be the indicator for the event that the labeling function $\lambda': E_S\to (0,1]$, mapping $e$ to its $\tau(e)$th sampled label has both of the following properties:
        \begin{enumerate} 
            \item $\max_{e_1,e_2 \in E_S}(\lambda'(e_1)-\lambda'(e_2))\leq \delta$ (i.e. all labels occur within a window of $\delta$)
            \item $\lambda'(e)$ is a valid realization of $P$.
        \end{enumerate}
        \item $\ContFull{\delta}_{S,\tau}=\LabEx_{S,\tau}\ContDel{\delta}_{S,\tau}$ (i.e. the indicator for the event that $(S,\tau)$ describes a genuine $\delta$-temporal $(H,P)$ motif in $\mathcal{G}$).
        \item $\ContSum{\delta}=\sum_{S \in \mathcal{S}, \tau \in \mathcal{T}_S}\ContFull{\delta}_{S,\tau}$ (i.e. the total number of appearances of $(H,P)$ as a $\delta$-temporal motif of $\mathcal{G}$)
    \end{itemize}
\end{definition}

Analogously, for the discrete case,
\begin{definition}
    For a random sample $\mathcal{H}$ from $\Gamma_{[n]}(\psi,T)$ (obtained by taking the $T$ discretization of $\mathcal{G}$), $S \in \mathcal{S}$ and $\tau \in \mathcal{T}_S$ we define:
    \begin{itemize}
        \item $\DiscDel{\beta}{T}_{S,\tau}$ to be the indicator for the event that the labeling function $\lambda': E_S\to (0,1]$ mapping $e$ to its $\tau(e)$th label has both of the following properties:
        \begin{enumerate}
            \item $\max_{e_1,e_2 \in E_S}(\lambda'(e_1)-\lambda'(e_2))< \beta$ (i.e. all labels fall within a window of $\beta$)
            \item $\lambda'(e)$ is a valid realization of $P$
        \end{enumerate}
        \item $\DiscFull{\beta}{T}_{S,\tau}=\LabEx_{S,\tau}\DiscDel{\beta}{T}_{S,\tau}$ (i.e. the indicator for the event that $(S,\tau)$ describes a genuine $\beta$-temporal motif of $\mathcal{H}$).
        \item $\DiscSum{\beta}{T}=\sum_{S \in \mathcal{S}, \tau \in \mathcal{T}_S}\DiscFull{\beta}{T}_{S,\tau}$ (i.e. the total number of appearances of $(H,P)$ as a $\beta$ temporal motif in $\mathcal{H}$).
    \end{itemize}
\end{definition}

With those definitions established, we can apply a lemma of \cite{Cliques}, suitably rephrased for our setting to get the probability of the event $\ContDel{\delta}_{S,\tau}$.
\begin{restatable}{lemma}{LemAlgowin}[Adapted from Lemma 3 of \cite{Cliques}]
    \label{lem: algowin}
    For $H=(V_H,E_H)$ a fixed non-trivial static graph, $P^{\emptyset}(E_H,\prec)$ the order over $E_H$ such that no two edges are comparable, $S \in \mathcal{S}$ and $\tau \in \mathcal{T}_S$,
    \[\Prob[\ContDel{\delta}_{S,\tau}=1]=\delta^{|E_H|}(1+\frac{|E_H|(1-\delta)}{\delta})\]
\end{restatable}
Now note that, whether or not a given set of labels all fall within a specific interval, is independent of the order they fall in. Thus we immediately obtain the following more general statement.
\begin{restatable}{corollary}{CorCont}
    \label{corr: cont}
    For $H=(V_H,E_H)$ a fixed non-trivial static graph and $P=(E_H,\prec)$ a partial order over its edges,
    \[\Prob[\ContDel{\delta}_{S,\tau}=1]=\frac{R(P)}{|E_H|!}\delta^{|E_H|}(1+\frac{|E_H|(1-\delta)}{\delta})\]
\end{restatable}
Unfortunately, the precise characterization of the discrete variant $\DiscDel{\beta}{T}_{S,\tau}$ is considerably less nice to work with.
It is in fact significantly more convenient (and as we will later see, sufficient) to work with the following loose bounds derived from the continuous case.
Consider that all labels that fall within an interval of length $\beta$ in $\Gamma_{[n]}(\psi,T)$ must fall within an interval of length $\frac{\beta}{T}$ in the associated instance of $\Gamma_{[n]}(\psi)$.
Formally, we make the following observation\footnote{We would like to note that this observation only holds for motifs obeying a strict partial order. The reason for this is that two labels which fall in an illegal order in the continuous sample may lie in the same snapshot after discretization and then may be permitted under the non-strict order. In order to consider a non-strict order, one can obtain a probabilistic variant of the observation, which should only affect our upper bounds by a constant (and therefore asymptotically irrelevant) factor.}.
\begin{restatable}{observation}{ObsBridgeTwo}
    \label{obs: bridge 2}
    For $\mathcal{G}=(V,E,\lambda)$, $T>\beta \geq 1$, $\mathcal{H}$ the $T$ discretization of $\mathcal{G}$, $H=(V_H,E_H)$ any static graph, $P=(E_H,\prec)$ a partial order over the edges of $H$, and for every $S \in \mathcal{S}$, $\tau \in \mathcal{T}_{S}$, we have that
    $\{\DiscFull{\beta}{T}_{S,\tau}=1\}\subseteq\{\ContFull{\frac{\beta}{T}}_{S,\tau}=1\}.$
    Therefore, with certainty,
    $\DiscSum{\beta}{T}\leq \ContSum{\frac{\beta}{T}}.$
\end{restatable}
We can now apply this alongside, \cref{lem: algowin}, \cref{corr: cont} and \cref{obs: bridge 1} to obtain bounds on $\DiscDel{\beta}{T}_{S,\tau}$.
\begin{restatable}{lemma}{LemBridgeThree}
    \label{lem: bridge 3}
    Let  $H=(V_H,E_H)$ be a fixed static graph, $P=(E_H,\prec)$ a partial order on its edges and $T\geq\beta\geq 1$. For any $S \in\mathcal{S}, \tau \in \mathcal{T}_S$,
    \[\Prob{[\DiscDel{\beta}{T}_{S,\tau}=1]}\leq \left(\frac{\beta}{T}\right)^{|E_H|}\left(1+\frac{|E_H|(1-\frac{\beta}{T})}{\frac{\beta}{T}}\right).\]
    Furthermore if $\beta>1$,
    \[\Prob{[\DiscDel{\beta}{T}_{S,\tau}=1]}\geq \left(\frac{\beta}{2T |E_H|}\right)^{|E_H|}\left(1+\frac{|E_H|(1-\frac{\beta}{2T})}{\frac{\beta}{2T}}\right).\]
    Finally, if $\beta=1$,
    \[\Prob{[\DiscDel{1}{T}_{S,\tau}=1]}=T^{1-|E_H|}.\]
\end{restatable}
\begin{proof}
    Let $X_{S,\tau}$ be the event that the labels of $\tau$ fall within an interval of length $\beta$ and let $Y_{S,\tau}$ be the event that the labels of $\tau$ obey $P$.
    Now consider a sample from $\Gamma_{[V]}(\psi)$, it follows from \cref{obs: bridge 1}, \cref{lem: algowin} and the fact that the set of intervals of the form $(\frac{i}{T},\frac{i+\beta}{T}]$ is a subset of the continuous intervals of length $\frac{\beta}{T}$ that,
    \[\Prob{[X_{S,\tau}]}\leq\left(\frac{\beta}{T}\right)^{|E_H|}\left(1+\frac{|E_H|(1-\frac{\beta}{T})}{\frac{\beta}{T}}\right).\]
    Now \[\Prob{[Y_{S,\tau}|X_{S,\tau}]}=\frac{R(P,\beta)}{\beta^{|E_H|}}.\]
    Combining the above with the trivial upper bound on $\frac{R(P,\beta)}{\beta^{|E_H|}}$, we have
    \[\Prob{[\DiscDel{\beta}{T}_{S,\tau}=1]}=\Prob{[Y_{S,\tau}|X_{S,\tau}]}\Prob{[X_{S,\tau}]}
    \leq \left(\frac{\beta}{T}\right)^{|E_H|}\left(1+\frac{|E_H|(1-\frac{\beta}{T})}{\frac{\beta}{T}}\right).\]

    For the lower bound, we consider sampling $\mathcal{H}$ first from $\Gamma_{V}(\psi)$ and then taking the $T$-discretization of our sample to obtain our dependent sample $\mathcal{G}$ from $\Gamma_{V}(\psi, \tau)$ in accordance with \cref{obs: bridge 1}.
    Now let $Z_{S,\tau}$ be the event that the labels of $\tau$ fall within an interval of length $\frac{\beta}{2T}$ in $\mathcal{H}$.
    From \cref{lem: algowin}, we have that
    \[\Prob{[Z_{S,\tau}]}=\left(\frac{\beta}{2T}\right)^{|E_H|}\left(1+\frac{|E_H|(1-\frac{\beta}{2T})}{\frac{\beta}{2T}}\right).\]
    However for any $\beta\geq 2$, we immediately obtain that any interval of length $\frac{\beta}{2T}$ must fall entirely within $\beta$ consecutive snapshots under a $T$-discretization.
    Thus in this case, \[\Prob{[X_{S,\tau}]}\geq \Prob{[Z_{S,\tau}]}.\]

    The bound is then obtained as in the upper bound, but taking the lower bound of $|E_H|^{-|E_H|}\leq \frac{R(P,\beta)}{\beta^{|E_H|}}$ obtained in \cref{lem: Poset-Extension}.

    Finally, for the case where $\beta=1$, there are $T$ possible snapshots and all labels land in each one with probability $T^{-|E_H|}$. Since the labels falling in each snapshot are mutually exclusive, we conclude with a tight union bound.
\end{proof}
\subsubsection{First Moment}
Since directly calculating the probability of a temporal motif occurring in a random temporal graph is infeasible, we shall rely on bounds derived from the first moment of the number of such occurrences. 
In fact, we are able to obtain the following lemma via a relatively standard calculation, exploiting both the linearity of expectation and the independence between the number of labels on each edge and their value.
Note that the term concerning the number of realizations of $P$ is absent in the case for $\beta=1$ as there exists only a single possible realization (all labels in a single snapshot).
\begin{restatable}{lemma}{LemExpectation}
    \label{lem: expectation}
    The expectation of $\ContSum{\delta}$ is:
    \[\Exp[\ContSum{\delta}]=\delta^{|E_H|}r^{|E_H|}|V_H|!\left(\frac{R(P)}{|E_H|!}\right)\binom{n}{|V_H|}\left(1+\frac{|E_H|(1-\delta)}{\delta}\right).\]
    Similarly, if $\beta=1$ then
    $\Exp{[\DiscSum{1}{T}]}=r^{|E_H|}T^{1-|E_H|}|V_H|!\binom{n}{|V_H|}$.\\ On the other hand, if $\beta>1$ then 
    $\Exp{[\DiscSum{\beta}{T}]}\geq \left(\frac{\beta r}{2|E_H|T}\right)^{|E_H|}|V_H|!\left(1+\frac{|E_H|(1-\frac{\beta}{2T})}{\frac{\beta}{2T}}\right)\binom{n}{|V_H|}.$ \medskip
\end{restatable}
\begin{proof}
    We present the proof for the continuous case, and defer the proof of the discrete cases to the appendices, as they follow from near identical calculations. 
    By linearity, we obtain,
    \[\Exp[\ContSum{\delta}]=\sum_{S \in \mathcal{S}}\sum_{\tau_S \in \mathcal{T}_S}\Exp[\LabEx_{S,\tau}\ContDel{\delta}_{S,\tau}].\]
    By independence we have that,
    \[=\sum_{S \in \mathcal{S}}\sum_{\tau \in \mathcal{T}_S}\Exp[\LabEx_{S,\tau}]\Exp[\ContDel{\delta}_{S,\tau}].\]
    From \cref{corr: cont} we have,
    \[=\delta^{|E_H|}\left(\frac{R(P)}{|E_H|!}\right)\left(1+\frac{|E_H|(1-\delta)}{\delta}\right)\sum_{S \in \mathcal{S}}\sum_{\tau \in \mathcal{T}_S}\Exp[\LabEx_{S,\tau}]\]
    By independence,
    \[=\delta^{|E_H|}\left(\frac{R(P)}{|E_H|!}\right)\left(1+\frac{|E_H|(1-\delta)}{\delta}\right)\sum_{S \in \mathcal{S}}\sum_{\tau \in \mathcal{T}_S}\prod_{e \in E_S}\Prob{\left[\tau(e) \leq l_e\right]}.\]
    Since $\sum_{\tau \in \mathcal{T}_S}\Prob{[\tau(e)\leq l_e]}=r$ is both finite and absolutely convergent, we observe that this is equivalent to the following via the $|E_H|$ fold Cauchy product,
    \[=\delta^{|E_H|}\left(\frac{R(P)}{|E_H|!}\right)\left(1+\frac{|E_H|(1-\delta)}{\delta}\right)\sum_{S \in \mathcal{S}}\prod_{e \in E_S}\sum_{\tau \in \mathcal{T}_S}\Prob{\left[\tau(e) \leq l_e\right]}.\]
    \[=\delta^{|E_H|}\left(\frac{R(P)}{|E_H|!}\right)\left(1+\frac{|E_H|(1-\delta)}{\delta}\right)\sum_{S \in \mathcal{S}}\prod_{e \in E_S} r\]
    \[=\delta^{|E_H|}r^{|E_H|}|V_H|!\left(\frac{R(P)}{|E_H|!}\right)\binom{n}{|V_H|}\left(1+\frac{|E_H|(1-\delta)}{\delta}\right).\]
\end{proof}
With that established, we may turn our attention to the construction of the actual bound. The observant reader may notice that the expectations obtained above will tend to zero whenever $\delta(n)=o(n^{-\rho_H})$ and $|V_H|$ is fixed. Infact \cref{lem: expectation} is sufficient to show the lower bounds of \cref{thm: Continuous Existence} and \cref{thm: Discrete Existence} via a straightforward application of the first moment method. Unfortunately, the same cannot be said for \cref{thm: Cliques}. Indeed, even a cursory inspection of the problem will reveal that $\ContFull{\delta}_{S_1,\tau_1}$ and $\ContFull{\delta}_{S_2,\tau_2}$ have a strong positive association whenever $\tau_1$ and $\tau_2$ share labels or even when $S_1$ and $S_2$ share a single edge. Consequently, the first moment produces an extremely weak bound that explodes whenever $r>1$ and the size of $H$ is unbounded. In order to control these dependencies and limit the contribution of additional occurrences after the first, we make use of a slightly more sophisticated bound based on the ratio of the expectation and the expectation conditioned on being non-zero.
\begin{restatable}{lemma}{LemMotifUpper}
    \label{lem: motif-upper}
    Let $H=(V_H,E_H)$ be a fixed graph, $P=(E_H,\prec)$, $\mathcal{G}=(V,E,\lambda)$ be a sample from $\Gamma_{[n]}(\psi)$ and $0<\delta<1$. Then,
    \[\Prob{\left[\ContSum{\delta}>0\right]}\leq \binom{n}{|V_H|}\frac{|V_{H}|!}{\# \text{Aut}(H,P)}\left(\frac{\delta rq}{(1-\delta)q+\delta r}\right)^{|E_H|}\left(1+\frac{|E_H|(1-\delta)}{\delta}\right).\]
\end{restatable}
\begin{proof}
    Fix $S \in \mathcal{S}$ and define by $\ContSum{\delta,S}=\sum_{\tau \in \mathcal{T}_S}\ContFull{\delta}_{S,\tau}$ the number of $\delta$-temporal cliques on $S$.
    We will begin with a pair of claims:
    \begin{claim}
        \[\Exp[\ContSum{\delta,S}|\ContSum{\delta,S}>0]\geq\Exp[\ContSum{\delta,S}|\ContFull{\delta}_{S,\tmn{S}}=1].\]
    \end{claim}
    \begin{proof}
    It follows that \[\Exp[\ContSum{\delta,S}|\ContSum{\delta,S}>0]\geq \min_{\hat{\tau} \in \mathcal{T}_S:\hat{\tau}\geq \tmn{S}}\Exp[\ContSum{\delta,S}|\ContFull{\delta}_{S,\hat{\tau}}=1]\]
    However, we still must determine $\hat{\tau}$.
    First, note that, for all $\tau_1,\tau_2 \in \mathcal{T}_S$, $\Prob{[\ContDel{\delta}_{S,\tau_1}=1|\ContDel{\delta}_{S,\tau_2}=1]}$ depends only on $\{e \in E_S:\tau_1(e)=\tau_2(e)\}$ (the set of edge labels shared by the two).
     Now for any $\tau_1,\tau_2, \tau_3 \in \mathcal{T}_S$ where $\tau_1,\tau_2\leq \tau_3$ we have, by the symmetry in the number of pairs of labelings sharing a given set of edge labels, that,
    \[\Exp[\ContSum{\delta,S}|\ContDel{\delta}_{S,\tau_1}=1\wedge \tmx{S}=\tau_3]=\Exp[\ContSum{\delta,S}=1|\ContDel{\delta}_{S,\tmn{S}}\wedge \tmx{S}=\tau_3]\]
    Note that this is monotonically increasing in $\tau_3$.
    Now,
    \[\Exp[\ContSum{\delta,S}|\ContFull{\delta}_{S,\tau}=1]=\Exp[\ContSum{\delta,S}|\ContDel{\delta}_{S,\tau}=1\wedge \tmx{S}\geq \tau]\]
    \[=\sum_{\tau_3 \in \mathcal{T}_S:\tau_3\geq \tau}\Exp[\ContSum{\delta,S}|\ContDel{\delta}_{S,\tau}=1\wedge \tmx{S}=\tau_3]\Prob{[\tmx{S}=\tau_3|\tmx{S}\geq \tau]}\]
    \[=\sum_{\tau_3 \in \mathcal{T}_S:\tau_3\geq \tau}\Exp[\ContSum{\delta,S}|\ContDel{\delta}_{S,\tmn{S}}=1\wedge \tmx{S}=\tau_3]\Prob{[\tmx{S}=\tau_3|\tmx{S}\geq \tau]}\]
    This is monotonically increasing in $\tau$ and thus we must have that it is minimized when $\tau$ is minimized (specifically when $\tau=\tmn{S}$). So we obtain the claim.
    \end{proof}
    
    With the identity of $\hat{\tau}$ established, we will now need a lower bound on $\Exp[\ContDel{\delta}_{S,\tau}|\ContDel{\delta}_{S,\tmn{S}}=1]$. To this end,
    \begin{claim}
        For all $S \in \mathcal{S}$ and $\tau \in \mathcal{T}_S$ 
        %such that $\min_{e \in E_S}(\tau(e))=1$
        ,
        \[\Exp{[\ContDel{\delta}_{S,\tau}|\ContDel{\delta}_{S,\tmn{S}}=1]}\geq \left(\frac{R(P)}{|E_H|!}\right)\prod_{e \in E_S:\tau(e)\neq 1}\delta\]
    \end{claim}
    \begin{proof}
    
    Consider the following reordering of sampling. First, the labels associated with every edge in ${e \in E_S:\tau(e)=1}$ are sampled, followed by both samples for every edge in ${e \in E_S:\tau(e)>1}$. Now note that with the set of shared labels fixed, $\ContDel{\delta}_{S,\tau}$ and $\ContDel{\delta}_{S,\tmn{S}}$ become independent Bernoulli random variables sharing the same success parameter determined by the positions of the shared labels. Denote by $X$ the $(0,1]$ random variable equal to the probability that all labels of $\tau$ will fall within an interval of $\delta$ and in an order satisfying $P$ once the shared labels have been sampled.
    Since $X$ has a defined cumulative distribution function and takes values only in $(0,1]$ its first and second moments must be finite.
    Thus, it follows from \cref{lem: dependent Bernoulli}, that
    \[\Prob[\ContDel{\delta}_{S,\tau}=1|\ContDel{\delta}_{S,\tmn{S}}=1]\geq \Prob[\ContDel{\delta}_{S,\tau}=1|X>0]\]
    Therefore, since $\ContDel{\delta}_{S,\tau}$ is supported only on $\{0,1\}$ we must have that, 
    \[\Exp[\ContDel{\delta}_{S,\tau}|\ContDel{\delta}_{S,\tmn{S}}=1]\geq \Exp[\ContDel{\delta}_{S,\tau}|X>0].\]
    Let $Y$ be an indicator for the event that all edge labels shared by $\tau$ and $\tmn{S}$ fall within an interval of length $\delta$. Then
    \[ \Exp[\ContDel{\delta}_{S,\tau}|X>0] \geq \Exp[\ContDel{\delta}_{S,\tau}|Y=1].\]
    The last expectation is trivially lower bounded by the probability that all labels yet to be sampled fall into a specific window of length $\delta$ and that all labels fall in a permitted order. Therefore,
    \[\Exp[\ContDel{\delta}_{S,\tau}|Y=1] \geq \left(\frac{R(P)}{|E_H|!}\right)\prod_{e \in E_S:\tau(e)\neq 1}\delta,\]
    which concludes the proof of the claim.
    \end{proof}
    Thus we have the claim and can now begin our main calculation,
    \[\Exp[\ContSum{\delta,S}|\ContSum{\delta,S}>0]\geq \Exp[\ContSum{\delta,S}|\ContFull{\delta}_{S,\tmn{S}}=1]\]
    \[= \sum_{\tau \in \mathcal{T}_S}\Exp[\ContFull{\delta}_{S,\tau}|\ContFull{\delta}_{S,\tmn{S}}=1]\]
    \[= \sum_{\tau \in \mathcal{T}_S}\Exp[\LabEx_{S,\tau}|\LabEx_{S,\tmn{S}}=1]\Exp[\ContDel{\delta}_{S,\tau}|\ContDel{\delta}_{S,\tmn{S}}=1]\]
    \[\geq \sum_{\tau \in \mathcal{T}_S}\Exp[\ContDel{\delta}_{S,\tau}|\ContDel{\delta}_{S,\tmn{S}}]\prod_{e \in E_S}\Prob[l_e\geq \tau(e)|l_e\geq 1]\]
    \[\geq \left(\frac{R(P)}{|E_H|!}\right)\sum_{\tau \in \mathcal{T}_S}\left(\prod_{e \in E_S}\Prob[l_e\geq \tau(e)|l_e\geq 1]\right)\left(\prod_{e\in E_S: \tau(e)\neq 1}\delta\right)\]
    For $Z\sim \psi$ a generic random variable,
    \[=\left(\frac{R(P)}{|E_H|!}\right)\sum_{x_1,..,x_{|E_H|}\geq 1}\left(\prod_{i \in [|E_H|]:x_i\neq 1}\delta\right)\left(\prod_{j=1}^{|E_H|}\Prob{[Z\geq x_j|Z\geq 1]}\right)\]
    By distributivity for absolutely convergent sums,
    \[=\left(\frac{R(P)}{|E_H|!}\right)\prod_{i \in [|E_H|]}\left[(1-\delta)+\delta\sum_{i\geq 1}\Prob{[Z\geq i|Z\geq 1]}\right]\]
    \[=\left(\frac{R(P)}{|E_H|!}\right)\prod_{i \in [|E_H|]}\left[(1-\delta)+\delta\Exp[Z|Z\geq 1]\right]\]
    \[=\left(\frac{R(P)}{|E_H|!}\right)\left(\frac{(1-\delta)q+\delta r}{q} \right)^{|E_H|}\]
    From the previous lemma we know that,
    \[\Exp[\ContSum{\delta,S}]=\delta^{|E_H|}r^{|E_H|}\left(\frac{R(P)}{|E_H|!}\right)\left(1+\frac{|E_H|(1-\delta)}{\delta}\right)\]
    So we now have that,
    \[\Prob{[\ContSum{\delta,S}>0]}= \frac{\Exp[\ContSum{\delta,S}]}{\Exp[\ContSum{\delta,S}|\ContSum{\delta,S}\geq 1]}\leq\frac{\delta^{|E_H|}r^{|E_H|}\left(1+\frac{|E_H|(1-\delta)}{\delta}\right)}{\left(\frac{(1-\delta)q+\delta r}{q} \right)^{|E_H|}}\]
    \[=\left(\frac{\delta rq}{(1-\delta)q+\delta r}\right)^{|E_H|}\left(1+\frac{|E_H|(1-\delta)}{\delta}\right)\]
    We then obtain the claim via a union bound over $\mathcal{S}$ combined with the observation that we over-count by at least a factor of $\#\text{Aut}(H,P)$.
\end{proof}

This result serves as our primary tool for obtaining both the lower bounds of \cref{thm: Continuous Existence} and \cref{thm: Discrete Existence}, as well as the upper bound for \cref{thm: Cliques}. It may come as a surprise that we present a result only for the continuous case, however due to the relationship between the models and the similarity of the thresholds we wish to obtain, it turns out that \cref{lem: motif-upper} is sufficient for our purposes.

\subsubsection{Second Moment}
Once again we wish to avoid a direct calculation of a lower bound on the probability of a given motif occurring, and so our strategy will be to use the second moment method.
Unfortunately, we do not find a unified overarching bound sufficient to obtain our results for both fixed and growing motifs, as positive association rears its head once more (although this time as an obstacle for the fixed case).
In particular, when two potential occurrences of a motif share a footprint they remain dependent even if they do not share any labels, due to the inherent randomness of the \distribution{}.
This presents an issue whenever $\psi$ is a non-degenerate non-Bernoulli distribution, as we will retain terms in our summations corresponding to such pairs.
These terms contribute an additional $\delta$ factor to our calculations, which can be absorbed in the case of \cref{thm: Cliques} where $\delta$ is constant but a naive application of the second moment would lead to asymptotically weaker bounds in the case of \cref{thm: Continuous Existence} and \cref{thm: Discrete Existence} where $\delta$ is decreasing.
In order to avoid this pitfall, we shall make use of two bounds, one for the case where $\psi$ is a Bernoulli random variable and one for the more general but still bounded case.
We shall warm up with the considerably easier former case.
\begin{restatable}{lemma}{LemVariance}
    \label{lem: variance}
    Let $\mathcal{G}=(V,E,\lambda)$ be a sample from $\Gamma_{[n]}(\psi)$, for $\psi$ a Bernoulli random variable, then for any $0<\delta<1$,
    \[\Var\left[\ContSum{\delta}\right]\leq(\sqrt{2}\delta r)^{2|E_H|}\left(1+|E_H|\frac{(1-\delta)}{\delta}\right)n^{2|V_H|}(2|V_H|)^{2|V_H|}\sum_{I\subseteq H}(2 \delta r)^{-|E_I|}n^{-|V_I|}.\]
    Similarly for any $\beta>0$, $T>\beta$,
    \[\Var\left[\DiscSum{\beta}{T}\right]\leq(\frac{\sqrt{2}\beta r}{T})^{2|E_H|}\left(1+|E_H|\frac{(T-\beta)}{\beta}\right)n^{2|V_H|}(2|V_H|)^{2|V_H|}\sum_{I\subseteq H}(\frac{2 \beta r}{T})^{-|E_I|}n^{-|V_I|}.\]
\end{restatable}
\begin{proof}
    We will present only the proof for the continuous case here, however it should be noted that the discrete case is identical replacing only the continuous events for their discrete analogues and $\delta$ with $\frac{\beta}{T}$.
    \[\Var\left[\ContSum{\delta}\right]=\Exp\left[\ContSum{\delta}^2\right]-\Exp\left[\ContSum{\delta}\right]^2\]
    \[=\sum_{S_1,S_2 \in \mathcal{S}}\Prob{\left[\ContFull{\delta}_{S_1,\tmn{S_1}}=1,\ContFull{\delta}_{S_2,\tmn{S_2}}=1\right]}-\Prob{\left[\ContFull{\delta}_{S_1,\tmn{S_1}}=1\right]}\Prob{\left[\ContFull{\delta}_{S_2,\tmn{S_2}}=1\right]}\]
    We only get contribution from terms where $S_1$ and $S_2$ share at least one edge.
    \[=\sum_{\substack{S_1,S_2 \in \mathcal{S}\\
    E_{S_1}\cap E_{S_2}\neq \emptyset}}\Prob{\left[\ContFull{\delta}_{S_1,\tmn{S_1}}=1,\ContFull{\delta}_{S_2,\tmn{S_2}}=1\right]}-\Prob{\left[\ContFull{\delta}_{S_1,\tmn{S_1}}=1\right]}\Prob{\left[\ContFull{\delta}_{S_2,\tmn{S_2}}=1\right]}\]
    Since all terms are non-negative,
    \[\leq\sum_{\substack{S_1,S_2 \in \mathcal{S}\\
    E_{S_1}\cap E_{S_2}\neq \emptyset}}\Prob{\left[\ContFull{\delta}_{S_1,\tmn{S_1}}=1,\ContFull{\delta}_{S_2,\tmn{S_2}}=1\right]}\]
    By independence,
    \[\leq\sum_{\substack{S_1,S_2 \in \mathcal{S}\\
    E_{S_1}\cap E_{S_2}\neq \emptyset}}\Prob{\left[\ContDel{\delta}_{S_1,\tmn{S_1}}=1,\ContDel{\delta}_{S_2,\tmn{S_2}}=1\right]}\Prob{\left[\LabEx_{S_1,\tmn{S_1}}=1,\LabEx_{S_2,\tmn{S_2}}=1\right]}\]
    Since $\psi$ is Bernoulli with expectation $r$,
    \[=\sum_{\substack{S_1,S_2 \in \mathcal{S}\\
    E_{S_1}\cap E_{S_2}\neq \emptyset}}\Prob{\left[\ContDel{\delta}_{S_1,\tmn{S_1}}=1,\ContDel{\delta}_{S_2,\tmn{S_2}}=1\right]}r^{|E_{S_1}\cup E_{S_2}|}\]
    We now need the following result,
    \begin{claim}
        For $k=|E_{S_2}\setminus E_{S_1}|$,
        \[\Prob{\left[\ContDel{\delta}_{S_1,\tmn{S_1}}=1,\ContDel{\delta}_{S_2,\tmn{S_2}}=1\right]}\leq 2^k\delta^{|E_H|+k}\left(1+|E_H|\frac{(1-\delta)}{\delta}\right)\]
    \end{claim}
    \begin{proof}
        \[\Prob{\left[\ContDel{\delta}_{S_1,\tmn{S_1}}=1,\ContDel{\delta}_{S_2,\tmn{S_2}}=1\right]}=\Prob{\left[\ContDel{\delta}_{S_2,\tmn{S_2}}=1|\ContDel{\delta}_{S_1,\tmn{S_1}}=1\right]}\Prob{\left[\ContDel{\delta}_{S_1,\tmn{S_1}}=1\right]}\]
        Applying, \cref{corr: cont} we have that,
        \[\leq\Prob{\left[\ContDel{\delta}_{S_2,\tmn{S_2}}=1|\ContDel{\delta}_{S_1,\tmn{S_1}}=1\right]}\delta^{|E_H|}\left(1+|E_H|\frac{(1-\delta)}{\delta}\right)\]
        Now, for $I$ the interval of length $2\delta$ centered on the minimal interval containing all labels shared between $\tmn{S_1}$ and $\tau$, we define the random variable $X$ to be the indicator random variable for the event that all remaining labels of $\tau$ fall within $I$.
        By construction $\{\ContDel{\delta}_{S_2,\tau}=1|\ContDel{\delta}_{S_1,\tmn{S_1}}=1\}\subseteq \{X=1|\ContDel{\delta}_{S_1,\tmn{S_1}}=1\}$ and $\Prob{[X=1|\ContDel{\delta}_{S_1,\tmn{S_1}}=1]}\leq (2\delta)^k$.
         Thus,
         \[\Prob{[\ContDel{\delta}_{S_1,\tmn{S_1}}=1,\ContDel{\delta}_{S_2,\tau}=1]}\leq (2\delta)^k\delta^{|E_H|}\left(1+\frac{|E_H|(1-\delta)}{\delta}\right)\]
         The claim follows.
    \end{proof}
    
    Returning to our main calculation and partitioning the sum over the isomorphism class of the intersection of the two motifs,
    \[=\sum_{I\subseteq H}\sum_{\substack{S_1,S_2 \in \mathcal{S}\\
    S_1\cap S_2\simeq I}}2^{|E_H|-|E_I|}\delta^{2|E_H|-|E_I|}r^{2|E_H|-|E_I|}\left(1+|E_H|\frac{(1-\delta)}{\delta}\right)\]
    \[=(\sqrt{2}\delta r)^{2|E_H|}\left(1+|E_H|\frac{(1-\delta)}{\delta}\right)\sum_{I\subseteq H}(2 \delta r)^{-|E_I|}|\{S_1,S_2 \in \mathcal{S}:S_1\cap S_2\simeq I\}|\]
    The final term is trivially upper bounded by $\binom{n}{2|V_H|-|V_I|}(2|V_H|)^{2|V_H|}\leq n^{2|V_H|-|V_I|}(2|V_H|)^{2|V_H|}$.
    \[=(\sqrt{2}\delta r)^{2|E_H|}\left(1+|E_H|\frac{(1-\delta)}{\delta}\right)n^{2|V_H|}(2|V_H|)^{2|V_H|}\sum_{I\subseteq H}(2 \delta r)^{-|E_I|}n^{-|V_I|}\]
    Which gives our claim.
\end{proof}

With that established, we can present the following more general result, in which we have to be much more careful to control the dependence caused by the label multiplicity distribution. As a consequence, the following lemma is restricted to the case where $\psi$ is almost surely bounded.
\begin{restatable}{lemma}{LemVarianceFull}
    \label{lem: variance-full}
    Let $\mathcal{G}=(V,E,\lambda)$ be a sample from $\Gamma_{[n]}(\psi)$, where $\psi$ is almost surely bounded, then for any $0<\delta<1$,
    \[\Var\left[\ContSum{\delta}\right]\leq \delta^{2|E_H|}r^{2|E_H|}\left(1+\frac{|E_H|(1-\delta)}{\delta}\right)^2\sum_{S_1,S_2 \in \mathcal{S}: E_{S_1}\cap E_{S_2}\neq \emptyset}\left(\frac{\delta r_2+r-\delta r}{\delta r^2}\right)^{|E_{S_1}\cap E_{S_2}|}.\]
\end{restatable}
\begin{proof}
    Denote $\mathcal{L}=\{(S,\tau):S \in \mathcal{S},\tau \in \mathcal{T}_S\}$ and consider the variance of $\ContSum{\delta}$,
    \[\Var[\ContSum{\delta}]=\Exp[\ContSum{\delta}^2]-\Exp[\ContSum{\delta}]^2\]
    \[=\sum_{(S_1,\tau_1),(S_2,\tau_2) \in \mathcal{L}} \Prob{[\ContFull{\delta}_{S_1,\tau_1}=1,\ContFull{\delta}_{S_2,\tau_2}=1]}-\Prob{[\ContFull{\delta}_{S_1,\tau_1}=1]}\Prob[\ContFull{\delta}_{S_2,\tau_2}=1]\]
    We only get contribution from terms where $\ContFull{\delta}_{S_1,\tau_1}$ and $\ContFull{\delta}_{S_2,\tau_2}$ are dependent. 
    In particular, where $S_1$ and $S_2$ share edges.
    \[=\sum_{S_1,S_2 \in \mathcal{S}: E_{S_1}\cap E_{S_2} \neq \emptyset}\sum_{\tau_1 \in \mathcal{T}_{S_1},\tau_2 \in \mathcal{T}_{S_2}}\Prob{[\ContFull{\delta}_{S_1,\tau_1}=1,\ContFull{\delta}_{S_2,\tau_2}=1]}-\Prob{[\ContFull{\delta}_{S_1,\tau_1}=1]\Prob{[\ContFull{\delta}_{S_2,\tau_2}=1]}}\]
    Since all terms are non-negative,
    \[\leq\sum_{S_1,S_2 \in \mathcal{S}: E_{S_1}\cap E_{S_2} \neq \emptyset}\sum_{\tau_1 \in \mathcal{T}_{S_1},\tau_2 \in \mathcal{T}_{S_2}}\Prob{[\ContFull{\delta}_{S_1,\tau_1}=1,\ContFull{\delta}_{S_2,\tau_2}=1]}\]
    Unfortunately, we have a high degree of dependence between possible occurrences on the same footprint even if they share no labels in common as the existence of these labels affects the probability of the other occurring.
    By the law of total probability, and combining the sums,
    \[= \sum_{S_1,S_2 \in \mathcal{S}: E_{S_1}\cap E_{S_2} \neq \emptyset}\sum_{\substack{\tau_1 \in \mathcal{T}_{S_1},\tau_2 \in \mathcal{T}_{S_2}\\ \tau^* \in \mathcal{T}_{S_1\cup S_2}}}\Prob{[\ContFull{\delta}_{S_1,\tau_1}=1,\ContFull{\delta}_{S_2,\tau_2}=1|\tmx{S_1\cup S_2}=\tau^*]}\Prob{[\tmx{S_1\cup S_2}=\tau^*]}\]
    By independence,
    \[=\sum_{S_1,S_2 \in \mathcal{S}: E_{S_1}\cap E_{S_2} \neq \emptyset}\sum_{\substack{\tau_1 \in \mathcal{T}_{S_1},\tau_2 \in \mathcal{T}_{S_2}\\ \tau^* \in \mathcal{T}_{S_1\cup S_2}}}\Prob{[\ContDel{\delta}_{S_1,\tau_1}=1,\ContDel{\delta}_{S_2,\tau_2}=1]}\Prob{[\tmx{S_1\cup S_2}=\tau^*]} \mathbbm{1}_{\tau_{S_1}\leq \tau^*}\mathbbm{1}_{\tau_{S_2}\leq \tau^*}\]
    By commutativity and distributivity,
    \[=\sum_{\substack{S_1,S_2 \in \mathcal{S}: E_{S_1}\cap E_{S_2} \neq \emptyset\\ \tau^* \in \mathcal{T}_{S_1\cup S_2}}}\Prob{[\tmx{S_1\cup S_2}=\tau^*]}\sum_{\tau_1 \in \mathcal{T}_{S_1},\tau_2 \in \mathcal{T}_{S_2}: \tau_1,\tau_2 \leq\tau^*}\Prob{[\ContDel{\delta}_{S_1,\tau_1}=1,\ContDel{\delta}_{S_2,\tau_2}=1]}\]
    We need an upper bound on the final term.
    Note that $\Prob{[\ContDel{\delta}_{S_1,\tau_1}=1,\ContDel{\delta}_{S_2,\tau_2}=1]}$ depends only on the edges shared between $\tau_1$ and $\tau_2$.
    Thus, by symmetry we obtain the following,
    \[=\sum_{\substack{S_1,S_2 \in \mathcal{S}: E_{S_1}\cap E_{S_2} \neq \emptyset\\ \tau^* \in \mathcal{T}_{S_1\cup S_2}}}\Prob{[\tmx{S_1\cup S_2}=\tau^*]}|\{\tau \in \mathcal{T}_{S_1}:\tau \leq \tau^*\}|\sum_{\tau \in \mathcal{T}_{S_2}: \tau \leq\tau^*}\Prob{[\ContDel{\delta}_{S_1,\tmn{S_1}}=1,\ContDel{\delta}_{S_2,\tau}=1]}\]
    We will need the following upper bound on the final summand,
    \begin{claim}\label{claim-1}
        For $k=|\{e \in E_{S_2}:\tau(e)\neq1\}|$,
        \[\Prob{[\ContDel{\delta}_{S_1,\tmn{S_1}}=1,\ContDel{\delta}_{S_2,\tau}=1]}\leq \delta^{|E_H|+k}\left(1+\frac{|E_H|(1-\delta)}{\delta}\right)^2\]
    \end{claim}
    \begin{proof}
        From \cref{corr: cont},
        \[\Prob{[\ContDel{\delta}_{S_1,\tmn{S_1}}=1,\ContDel{\delta}_{S_2,\tau}=1]}\leq \delta^{|E_H|}\left(1+\frac{|E_H|(1-\delta)}{\delta}\right)\Prob{[\ContDel{\delta}_{S_2,\tau}=1|\ContDel{\delta}_{S_1,\tmn{S_1}}=1]}\]
        We consider a relaxation of the event $\{\ContDel{\delta}_{S_2,\tau}=1|\ContDel{\delta}_{S_1,\tmn{S_1}}=1\}$. 
        Let $Y$ be the indicator random variable for the event that all remaining labels of $\tau$ fall within $\delta$ of each other. It follows from \cref{lem: algowin}, that $\Prob{[Y|\ContDel{\delta}_{S_1,\tmn{S_1}}=1]}\leq \delta^{k}\left(1+\frac{|E_H|(1-\delta)}{\delta}\right)$.
        Thus,
        \[\Prob{[\ContDel{\delta}_{S_1,\tmn{S_1}}=1,\ContDel{\delta}_{S_2,\tau}=1]}\leq \delta^{|E_H|+k}\left(1+\frac{|E_H|(1-\delta)}{\delta}\right)^2\]
        Which gives the claim.
    \end{proof}
    
    Applying \cref{claim-1} we find that,
    \[\sum_{\tau \in \mathcal{T}_{S_2}: \tau \leq\tau^*}\Prob{[\ContDel{\delta}_{S_1,\tmn{S_1}}=1,\ContDel{\delta}_{S_2,\tau}=1]}\leq \sum_{i=0}^{|E_{S_1}\cap E_{S_2}|}\sum_{\substack{\tau \in \mathcal{T}_{S_2}: \tau \leq\tau^* \\
     |\{e \in E_{S_1\cap S_2}: \tau(e)=1\}|=i}}\delta^{2|E_H|-i}\left(1+\frac{|E_H|(1-\delta)}{\delta}\right)^2\]
     \[=\delta^{2|E_H|}\left(1+\frac{|E_H|(1-\delta)}{\delta}\right)^2 \sum_{i=0}^{|E_{S_1}\cap E_{S_2}|}\sum_{\substack{\tau \in \mathcal{T}_{S_2}: \tau \leq\tau^* \\
     |\{e \in E_{S_1\cap S_2}: \tau(e)=1\}|=i}}\delta^{-i}\]
     \[=\delta^{2|E_H|}\left(1+\frac{|E_H|(1-\delta)}{\delta}\right)^2\left(\prod_{e \in E_{S_1}\cap E_{S_2}}\left(\delta^{-1}+\tau^*(e)-1\right)\right)\left(\prod_{e \in E_{S_2} \setminus E_{S_1}} \tau^*(e)\right)\]
     Returning to our main calculation,
    \begin{align*}
        \Var[\ContSum{\delta}]\leq& \delta^{2|E_H|}\left(1+\frac{|E_H|(1-\delta)}{\delta}\right)^2\sum_{\substack{S_1,S_2 \in \mathcal{S}: E_{S_1}\cap E_{S_2} \neq \emptyset\\ \tau^* \in \mathcal{T}_{S_1\cup S_2}}}\left[\begin{aligned}&\Prob{[\tmx{S_1\cup S_2}=\tau^*]}|\{\tau \in \mathcal{T}_{S_1}:\tau \leq \tau^*\}|\\&\cdot\left(\prod_{e \in E_{S_1\cap S_2}}\left(\delta^{-1}+\tau^*(e)-1\right)\right)\\&\cdot\left(\prod_{e \in E_{S_2 \setminus S_1}} \tau^*(e)\right)\end{aligned}\right]\\
        =&\delta^{2|E_H|}\left(1+\frac{|E_H|(1-\delta)}{\delta}\right)^2\sum_{\substack{S_1,S_2 \in \mathcal{S}: E_{S_1}\cap E_{S_2} \neq \emptyset\\ \tau^* \in \mathcal{T}_{S_1\cup S_2}}}\left[\begin{aligned}&\Prob{[\tmx{S_1\cup S_2}=\tau^*]}\left(\prod_{e \in E_{S_1\setminus S_2}}\tau^*(e)\right)\\&\cdot\left(\prod_{e \in E_{S_2 \setminus S_1}} \tau^*(e)\right)\\ &\cdot\left(\prod_{e \in E_{S_1\cap S_2}}\tau^*(e)\left(\delta^{-1}+\tau^*(e)-1\right)\right)\end{aligned}\right]
    \end{align*}
    Now by independence, for $Z \sim \psi$,
    \begin{align*}
        =\delta^{2|E_H|}\left(1+\frac{|E_H|(1-\delta)}{\delta}\right)^2\sum_{\substack{S_1,S_2 \in \mathcal{S}: E_{S_1}\cap E_{S_2} \neq \emptyset\\ \tau^* \in \mathcal{T}_{S_1\cup S_2}}}
        \left[\begin{aligned}&\left(\prod_{e \in E_{S_1\setminus S_2}}\tau^*(e)\Prob{[Z=\tau^*(e)]}\right)\\
        \cdot&\left(\prod_{e \in E_{S_1\cap S_2}}\tau^*(e)\left(\delta^{-1}+\tau^*(e)-1\right)\Prob{[Z=\tau^*(e)]}\right)\\
        \cdot&\left(\prod_{e \in E_{S_2 \setminus S_1}} \tau^*(e)\Prob{[Z=\tau^*(e)]}\right)\end{aligned}\right]
    \end{align*}
    From the definition of $\tau^*$ and distributivity,
    \begin{align*}
        =&\delta^{2|E_H|}\left(1+\frac{|E_H|(1-\delta)}{\delta}\right)^2\sum_{S_1,S_2 \in \mathcal{S}: E_{S_1}\cap E_{S_2} \neq \emptyset}\left[\begin{aligned}&\left(\prod_{e \in E_{S_1\setminus S_2}}\sum_{z\geq 0}z\Prob{[Z=z]}\right)\left(\prod_{e \in E_{S_2 \setminus S_1}} \sum_{z\geq 0} z\Prob{[Z=z]}\right)\\\cdot&\left(\prod_{e \in E_{S_1\cap S_2}}\sum_{z\geq 0}z\left(\delta^{-1}+z-1\right)\Prob{[Z=z]}\right) \end{aligned}\right]\\
        =&\delta^{2|E_H|}\left(1+\frac{|E_H|(1-\delta)}{\delta}\right)^2\sum_{S_1,S_2 \in \mathcal{S}: E_{S_1}\cap E_{S_2} \neq \emptyset}r^{2(|E_H|-|E_{S_1}\cap E_{S_2}|)}\left(\frac{\delta r_2+r-\delta r}{\delta}\right)^{|E_{S_1}\cap E_{S_2}|}\\
        =&\delta^{2|E_H|}r^{2|E_H|}\left(1+\frac{|E_H|(1-\delta)}{\delta}\right)^2\sum_{S_1,S_2 \in \mathcal{S}: E_{S_1}\cap E_{S_2}\neq \emptyset}
        \left(\frac{\delta r_2+r-\delta r}{\delta r^2}\right)^{|E_{S_1}\cap E_{S_2}|}
    \end{align*}
\end{proof}
\subsection{Fixed Motifs}
We can now turn our attention to the existence thresholds for fixed motifs specifically. The upper bound for the continuous case is an immediate consequence of \cref{lem: motif-upper}, as the upper bound it provides tends to $0$ whenever $H$ and $\psi$ are fixed, but $\delta=o(n^{-\rho_H})$. One should note here that while the \distribution{} and the partial ordering over $E_H$ affect the absolute upper bound on the probability of a given motif appearing, they do not affect its asymptote and thus do not appear in the threshold.
\begin{restatable}{lemma}{LemContLower}
    \label{corr: CONT-LOWER}
    Let $H=(V_H,E_H)$ be a fixed non-trivial static graph, $P=(E_H,\prec)$ a partial order on its edges, and $\mathcal{G}=(V,E,\lambda)$ a random sample from $\Gamma_{[n]}(\psi)$, then, for $\delta=o(n^{-\rho_H})$, $\mathcal{G}$ does not contain $(H,P)$ as a $\delta$-temporal motif with high probability.
\end{restatable}
\begin{proof}
    From \cref{lem: motif-upper},
    \[\Prob{\left[\ContSum{\delta}>0\right]}\leq \binom{n}{|V_H|}\frac{|V_{H}|!}{\# \text{Aut}(H,P)}\left(\frac{\delta rq}{(1-\delta)q+\delta r}\right)^{|E_H|}\left(1+\frac{|E_H|(1-\delta)}{\delta}\right)\]
    \[\leq \binom{n}{|V_H|}|V_H|!\Theta(\delta^{|E_H|-1})=O(n^{|V_H|}\delta^{|E_H|-1})\]
    This is $o(1)$ whenever $\delta(n)=o(n^{-\frac{|V_H|}{|E_H|-1}})$.
    Now consider $I \sqsubseteq H$ minimizing $\frac{|V_I|}{|E_I|-1}$ and $P_I$ the restriction of $P$ to $I$.
    Since $(I,P_I)$ does not appear in $\mathcal{G}$ as a $\delta$-temporal motif whenever $\delta(n)=o(n^{-\rho_H})$ and $(I,P_I)$ must appear for $(H,P_H)$ to appear, we obtain the result.
\end{proof}
Now, while we cannot apply \cref{lem: motif-upper} to the discrete case directly, we are able to exploit the relationship between the two models in order to obtain an analogous result. In particular, we know from \cref{obs: bridge 2} that an occurrence of a $\beta$-temporal motif in the $T$ discretization of $\mathcal{G}$, sampled from $\Gamma_{[n]}(\psi)$, is at most as probable as the occurrence of a $\frac{\beta}{T}$-temporal motif in $\mathcal{G}$. We therefore obtain the following, as an immediate corollary of \cref{corr: CONT-LOWER}.
\begin{restatable}{corollary}{CorrDiscLower}
    \label{corr: DISC-LOWER}
    Let $H=(V_H,E_H)$ be a fixed non-trivial static graph, $P=(E_H,\prec)$ a partial order on its edges, $\beta: \mathbb{N}\to \mathbb{N}$, $T:\mathbb{N}\to \mathbb{N}$ and $\mathcal{H}=(V,E,\lambda)$ a random sample from $\Gamma_{[n]}(\psi,T)$. Then for $\frac{\beta(n)}{T(n)}=o(n^{-\rho_H})$, $\mathcal{H}$ does not contain $(H,P)$ as a $\beta$-temporal motif with high probability.
\end{restatable}
Together these provide the lower bounds on the existence threshold presented in \cref{thm: Continuous Existence} and \cref{thm: Discrete Existence} respectively. The upper bounds are slightly more complex to obtain. Our basic strategy is an application of the second moment method using our bounds on the variance obtained in \cref{lem: variance}. Unfortunately, due to the obstacles brought about by dependence discussed previously, \cref{lem: variance} only holds when $\psi$ is a Bernoulli random variable. The naive application of the second moment, then, yields the following pair of results (where we defer the near identical proof of the discrete version) which are both considerably weaker than the more general statements we require for \cref{thm: Continuous Existence} and \cref{thm: Discrete Existence}.
\begin{restatable}{lemma}{LemWeakContUpper}
    \label{corr: Weak-cont-upper}
    Let $H=(V_H,E_H)$ be a fixed non-trivial static graph, $P=(E_H,\prec)$ a partial order on its edges, and $\mathcal{G}=(V,E,\lambda)$ a random sample from $\Gamma_{[n]}(\psi)$, such that $\psi$ is a Bernoulli distribution, then for $\delta=\omega(n^{-\rho_H})$, $\mathcal{G}$ contains $(H,P)$ as a $\delta$-temporal motif with high probability.
\end{restatable}
\begin{restatable}{lemma}{LemWeakDiscUpper}
    \label{corr: Weak-disc-upper}
    Let $H=(V_H,E_H)$ be a fixed non-trivial static graph, $P=(E_H,\prec)$ a partial order on its edges, and $\mathcal{H}=(V,E,\lambda)$ a random sample from $\Gamma_{[n]}(\psi,T)$, such that $\psi$ is a Bernoulli distribution, then for $\frac{\beta}{T}=\omega(n^{\frac{-|V_H|}{|E_H|-1}})$ $\mathcal{H}$ contains $(H,P)$ as a $\delta$-temporal motif with high probability.
\end{restatable}
\begin{proof}[Proof of \cref{corr: Weak-cont-upper}]
    \[\Prob{\left[\ContSum{\delta}=0\right]\leq \frac{\Var{\left[\ContSum{\delta}\right]}}{\Exp\left[\ContSum{\delta}\right]^2}}\]
    From \cref{lem: variance} and \cref{lem: expectation},
    \[\leq \frac{(\sqrt{2}\delta r)^{2|E_H|}\left(1+|E_H|\frac{(1-\delta)}{\delta}\right)n^{2|V_H|}(2|V_H|)^{2|V_H|}\sum_{I\subseteq H}(2 \delta r)^{-|E_I|}n^{-|V_I|}}{\left(\delta^{|E_H|}r^{|E_H|}|V_H|!\left(\frac{R(P)}{|E_H|!}\right)\binom{n}{|V_H|}\left(1+\frac{|E_H|(1-\delta)}{\delta}\right)\right)^2}\]
    Since $H$ is fixed, applying \cref{lem: Poset-Extension} and using $\binom{a}{b}\geq \left(\frac{a}{b}\right)^b$,
    \[\leq O(1)\sum_{I\subseteq H}\frac{(\delta)^{-|E_I|}n^{-|V_I|}}{\left(1+\frac{|E_H|(1-\delta)}{\delta}\right)}\]
    Now observe that all terms associated with trivial subgraphs of $H$ contribute $o(1)$ and thus,
    \[=O(\max_{I \sqsubseteq H}n^{-|V_I|}\delta^{1-|E_I|})+o(1)\]
    Since $\delta=\omega(\max_{I\sqsubseteq H}(n^{\frac{-|V_I|}{|E_I|-1}}))$,
    \[=o(1)\]
\end{proof}
Despite this weakness, however, we are still able to extract the results we need. The basic observation is as follows, for any valid choice of $\psi$, there exists a Bernoulli random variable $\sigma$ stochastically dominated by $\psi$. However, we also know from \cref{lem: domination} that the number of copies of our temporal motif in a sample from $\Gamma_{[n]}(\psi)$ stochastically dominates the number of copies in a sample from $\Gamma_{[n]}(\sigma)$. Together with \cref{corr: Weak-cont-upper} and \cref{corr: Weak-disc-upper}, this is sufficient to obtain the following two results.
\begin{restatable}{lemma}{LemContUpper}
    \label{corr: CONT-UPPER}
    Let $H=(V_H,E_H)$ be a fixed non-trivial graph, $P=(E_H,\prec)$ a partial order on its edges, and $\mathcal{G}=(V,E,\lambda)$ a random sample from $\Gamma_{[n]}(\psi)$, then for $\delta=\omega(n^{\frac{-|V_H|}{|E_H|-1}})$ $\mathcal{G}$ contains $(H,P)$ as a $\delta$-temporal motif with high probability.
\end{restatable}
\begin{proof}
    Since $\psi$ is a discrete non-negative distribution which is positive with positive probability, there must exist a Bernoulli random variable $\sigma$ such that $\psi$ stochastically dominates $\sigma$.
    Furthermore, it follows from \cref{corr: Weak-cont-upper}, that $(H,P)$ exists in $\mathcal{I}=(V,E,\lambda')$ with high probability, where $\mathcal{I}$ is sampled from $\Gamma_{[n]}(\sigma)$.
    From \cref{lem: domination}, it thus follows that $(H,P)$ exists within $\mathcal{G}$ with high probability.
\end{proof}
The following discrete version follows from an essentially identical argument.
\begin{restatable}{lemma}{LemDiscUpper}
    \label{corr: DISC-UPPER}
    Let $H=(V_H,E_H)$ be a fixed non-trivial static graph, $P=(E_H,\prec)$ a partial order on its edges, and $\mathcal{H}=(V,E,\lambda)$ a random sample from $\Gamma_{[n]}(\psi,T)$, then for $\frac{\beta}{T}=\omega(n^{\frac{-|V_H|}{|E_H|-1}})$, $\mathcal{H}$ contains $(H,P)$ as a $\delta$-temporal motif with high probability.
\end{restatable}
\cref{thm: Continuous Existence} and \cref{thm: Discrete Existence} then follow immediately from \cref{corr: CONT-LOWER} and \cref{corr: CONT-UPPER}, and \cref{corr: DISC-LOWER} and \cref{corr: DISC-UPPER}, respectively.
\subsection{Large Cliques}
    Finally, we shall consider the case where the temporal motifs grow with the size of the temporal graph but $\delta$ the size of the window remains constant. For the sake of brevity, we only consider the continuous model. However, one should be able to derive very similar results in the discrete case via the same method. Similarly, for comparison with previous results, we shall focus on the size of the largest motif with a clique footprint and no restriction on the ordering of its edges. We begin with an upper bound on the size of such a motif, which similarly to the case of fixed motifs, we obtain as a consequence of \cref{lem: motif-upper}.
\begin{restatable}{lemma}{LemCliquesUpper}
    \label{lem: cliques upper}
    Let $\mathcal{G}=(V,E,\lambda)$ be a sample from $\Gamma_{[n]}(\psi)$. Then for any $0<\epsilon<\delta<1$, $\mathcal{G}$ does not contain an occurrence of $(K_k,P^{\emptyset}_k)$ as a $\delta$-temporal motif with high probability, where \[k=\frac{2(1+\epsilon)\log{n}}{\log{\left(\frac{\delta r +q(1-\delta) }{q\delta r}\right)}}\]
\end{restatable}
\begin{proof}
    Consider $H=K_k$ the clique of size $k$, from \cref{lem: motif-upper},
    \[\Prob{\left[\ContSum{\delta}>0\right]}\leq \binom{n}{|V_H|}\frac{|V_{H}|!}{\# \text{Aut}(H,P)}\left(\frac{\delta rq}{(1-\delta)q+\delta r}\right)^{|E_H|}\left(1+\frac{|E_H|(1-\delta)}{\delta}\right)\]
    Since $H$ is a clique and $P$ is trivial,
    \[\leq n^{k}\left(\frac{\delta rq}{(1-\delta)q+\delta r}\right)^{\binom{k}{2}}\left(1+\frac{\binom{k}{2}(1-\delta)}{\delta}\right)\]
    Taking $k=(1+\epsilon)\frac{2\log{n}}{\log{\left(\frac{\delta r +q(1-\delta) }{q\delta r}\right)}}$,
    \[=\exp\left[(1+\epsilon)\frac{2\log^2{n}}{\log{\left(\frac{\delta r +q(1-\delta) }{q\delta r}\right)}}-(1+\epsilon)^2\frac{2\log^2{n}}{\log{\left(\frac{\delta r +q(1-\delta) }{q\delta r}\right)}}+o(\log^2 n)\right]\]
    \[=o(1)\]
    Thus we have the claim.
\end{proof}
Unfortunately, the argument used for \cref{corr: CONT-UPPER} is not sufficiently precise for our purposes here, as both the first and second moment of $\psi$ will have a significant impact on the threshold. Instead we make use of \cref{lem: variance-full}, which is better suited for this context as we need to account for the role of $r$ but can absorb the cost of the constant overhead introduced by an additional $\frac{1}{\delta}$ factor. For the case of almost surely bounded random variables, we then obtain the following from \cref{lem: variance-full} via a straightforward application of the second moment method.
\begin{restatable}{lemma}{LemCliquesBounded}
    \label{lem: cliques bounded}
    For any $0<\epsilon<\delta<1$, $\psi$ a valid and almost surely bounded \distribution{} and $\mathcal{G}=(V,E,\lambda)$ a sample from $\Gamma_{[n]}(\psi)$, we have that $\mathcal{G}$ contains $(K_{k},P^{\emptyset}_{k})$ as a $\delta$-temporal motif with high probability, where 
    \[k=(1-\epsilon)\frac{2\log{n}}{\log{\left(\frac{\delta r_2-\delta r+r}{\delta r^2}\right)}}\]
\end{restatable}
\begin{proof}
    From \cref{lem: variance-full} we know that,
    \[\Var\left[\ContSum{\delta}\right]\leq \delta^{2|E_H|}r^{2|E_H|}\left(1+\frac{|E_H|(1-\delta)}{\delta}\right)^2\sum_{S_1,S_2 \in \mathcal{S}: E_{S_1}\cap E_{S_2}\neq \emptyset}\left(\frac{\delta r_2+r-\delta r}{\delta r^2}\right)^{|E_{S_1}\cap E_{S_2}|}\]
    Since $H$ is a clique we obtain,
    \[=\delta^{k(k-1)}r^{k(k-1)}\left(1+\frac{\binom{k}{2}(1-\delta)}{\delta}\right)^2(k!)^2\binom{n}{k}\sum_{i=2}^k\binom{n-k}{k-i} \binom{k}{i}\left(\frac{\delta r_2+r-\delta r}{\delta r^2}\right)^{\binom{i}{2}}\]
    \begin{claim}
        For $n$ sufficiently large, $W>1$, $0<\epsilon<1$ and $k=\frac{(1-\epsilon)\log{n}}{\log{W}}$,
        \[f(i)=\binom{n-k}{k-i}\binom{k}{i}W^{i^2-i}\]
        is maximized over $\{2,...,\lfloor k\rfloor\}$ at $2$.
    \end{claim}
    \begin{proof}
        For $i>2$, let $g(i)=\frac{f(2)}{f(i)}$ then,
        \[g(i)=\left(\frac{(k-i)!}{(k-2)!}\right)^2\frac{i!}{2}\frac{(n-2k+i)!}{(n-2k+2)!}W^{2+i-i^2}\]
        \[\geq  (k-2)^{4-2i}(n-2k+2)^{i-2}W^{2+i-i^2}\]
        \[=\left(\frac{n-2k+2}{(k-2)^2W^{i+1}}\right)^{i-2}\]
        For the sake of contradiction assume that $f$ is not maximized over $\{2,...,\lfloor k \rfloor\}$ at $2$. Thus there must be some value of $2<j\leq k$ for which $g(j)<1$.
        \[\left(\frac{n-2k+2}{(k-2)^2W^{j+1}}\right)^{j-2}<1\]
        \[\left(\frac{n-2k+2}{(k-2)^2}\right)<W^{j+1}\]
        \[\log{(n-2k+2)}-2\log{(k-2)}<(j+1)\log{W}\]
        \[j>\frac{\log{(n-2k+2)}-2\log{(k-2)}}{\log{W}}-1\]
        Since $k=\Theta(\log{n})$,
        \[j>\frac{\log{n}}{\log{W}}-o(\log{n})\]
        However, since $j\leq k=\frac{(1-\epsilon)\log{n}}{\log{W}}$ we obtain a contradiction.
    \end{proof}
    Thus applying the claim,
    \[\Var{\left[\ContSum{\delta}\right]}\leq k(k!)^2\binom{n}{k}\binom{n-k}{k-2} \binom{k}{2}(\delta r)^{k(k-1)}\left(1+\frac{\binom{k}{2}(1-\delta)}{\delta}\right)^2\left(\sqrt{\frac{\delta r_2+r-\delta r}{\delta r^2}}\right)\]
    Now,
    \[\Prob{[\ContSum{\delta}=0]}\leq \frac{\Var{[\ContSum{\delta}]}}{\Exp[\ContSum{\delta}]^2}\]
    From the above, \cref{lem: expectation}, and the fact that all elements of $P^{\emptyset}_{k}$ are incomparable.
    \[\leq \frac{k(k!)^2\binom{n}{k}\binom{n-k}{k-2} \binom{k}{2}(\delta r)^{k(k-1)}\left(1+\frac{\binom{k}{2}(1-\delta)}{\delta}\right)^2\left(\sqrt{\frac{\delta r_2+r-\delta r}{\delta r^2}}\right)}{\left((\delta r)^{\binom{k}{2}}\binom{n}{k}(k!)\left(1+\frac{\binom{k}{2}(1-\delta)}{\delta}\right)\right)^2}\]
    \[=\frac{k\left(\sqrt{\frac{\delta r_2+r-\delta r}{\delta r^2}}\right)\binom{n-k}{k-2}}{\binom{n}{k}}\]
    \[=o(1)\]
    Which gives the claim.
\end{proof}
However, this result on its own is insufficient for the proof of \cref{thm: Cliques} in which we place no condition on the boundedness of the \distribution. Fortunately, we can exploit a similar trick to the case of fixed motifs and extend this result to cover unbounded \distribution s via a domination argument. In particular, for any valid choice of $\psi$, we are able to obtain a second almost surely bounded distribution $\sigma$ such that: i) $\psi$ stochastically dominates $\sigma$ and ii) the size of the largest clique in a sample from $\Gamma_{[n]}(\sigma)$ is very close to that in a sample from $\Gamma_{[n]}(\psi)$.
\begin{restatable}{lemma}{LemCliquesLower}
    \label{lem: cliques lower}
    For any $0<\epsilon<\delta<1$, $\psi$ a valid \distribution{} and $\mathcal{G}=(V,E,\lambda)$ a sample from $\Gamma_{[n]}(\psi)$. Then, $\mathcal{G}$ contains $(K_{k},P^{\emptyset}_{k})$ as a $\delta$-temporal motif with high probability, where 
    \[k=(1-\epsilon)\frac{2\log{n}}{\log{\left(\frac{\delta r_2-\delta r+r}{\delta r^2}\right)}}\]
\end{restatable}
\begin{proof}
    Our goal shall be to lower bound the probability of a $\delta$-temporal clique of size $k$ in a sample from $\Gamma_{[n]}(\psi)$ with the probability of one occurring in a sample from $\Gamma_{[n]}(\sigma)$, where $\sigma$ is a sufficiently close, but almost surely bounded distribution.

    To begin with, observe that since $\psi$ is discrete, non-negative and has finite first and second moments, the following two sums converge absolutely from below. For $X \sim \psi$,
    \[E[X]=\sum_{x\geq 0}x \Prob{[X=x]}\]
    \[E[X^2]=\sum_{x\geq 0}x^2 \Prob{[X=x]}\]
    Further, since $\Exp[X]>0$, from the definition of convergence we must have the following: for any $\alpha>0$ there exists $c$ such that,
    \[\sum_{0\leq x\leq c}x \Prob{[X=x]}\geq (1-\alpha)\Exp[X]\]
    \[\sum_{0\leq x\leq c}x^2 \Prob{[X=x]}\geq (1-\alpha)\Exp[X^2]\]
    For some value of $\alpha$ that we will fix later, we shall take $\sigma$ to be the distribution obtained by truncating $\psi$ to the interval $[0,c]$.
    We immediately obtain all of the following,
    \begin{itemize}
        \item $\sigma$ is almost surely bounded
        \item $\psi$ stochastically dominates $\sigma$
        \item For $X\sim \psi$ and $Y \sim \sigma$,
        \[(1-\alpha)\Exp[X]\leq \Exp[Y]\leq \Exp[X]\]
        \[(1-\alpha)\Exp[X^2]\leq \Exp[Y^2]\leq \Exp[X^2]\]
    \end{itemize}
    It, then, follows from \cref{lem: cliques bounded}, that for any choice of $\beta>0$, with high probability, a sample from $\Gamma_{[n]}(\sigma)$ contains $(K_{k'},P^{\emptyset}_{k'})$ as a $\delta$-temporal motif, where \[k'=\frac{2(1-\beta)\log{n}}{\log{\left(\frac{\delta\Exp[Y^2]+\Exp[Y]-\delta \Exp[Y]}{\delta \Exp[Y^2]}\right)}}\]
    However, from the properties of $\sigma$ and monotonicity of the logarithm,
    \[\log{\left(\frac{\delta\Exp[Y^2]+\Exp[Y]-\delta \Exp[Y]}{\delta \Exp[Y]^2}\right)}\leq \log{\left(\frac{\delta\Exp[X^2]+\Exp[X]-\delta\Exp[X]}{(1-\alpha)^2\delta \Exp[X]^2}\right)}\]
    Now for any choice of $\gamma>0$, there exists a choice of $\alpha>0$ sufficiently small such that,
    \[\leq (1+\gamma)\log{\left(\frac{\delta\Exp[X^2]+\Exp[X]-\delta \Exp[X]}{\delta \Exp[X]^2}\right)}\]
    Taking $\beta$ and $\gamma$ such that $\frac{1-\beta}{1+\gamma}=1-\epsilon$, we have $k'\geq k$, and so with high probability a sample from $\Gamma_{[n]}(\sigma)$ contains $(K_k,P^{\emptyset}_{k})$ as a $\delta$-temporal motif.

    With that established, it only remains to apply \cref{lem: domination}. Since $\psi$ stochastically dominates $\sigma$, we must have that the probability of $(K_k,P^{\emptyset}_{k})$ occurring as a $\delta$-temporal motif in a sample from $\Gamma_{[n]}(\psi)$ is at least the probability of $(K_k,P^{\emptyset}_{k})$ occurring as a $\delta$-temporal motif in a sample from $\Gamma_{[n]}(\sigma)$.
    Thus, the claim follows.
\end{proof}
\cref{thm: Cliques} then follows immediately from \cref{lem: cliques upper} and \cref{lem: cliques lower}.
\section{Doubling Times}
In this section, we outline our proof of \cref{thm: continuous double} characterizing the doubling time of $\Gamma_{[n]}(\psi)$ whenever $\psi$ is a degenerate random variable. The proof largely consists of obtaining probabilistic bounds on the doubling time of small individual sets, followed by a more careful analysis of the particularly early and late participants. This approach is inspired by~\cite{SharpThresholds}, although as we deal with relatively large sets of vertices rather than individual vertices, we have to put in a lot of extra-leg work in the process.

For the sake of mathematical convenience, we first establish the results for a slightly different random temporal graph model and then show how they can be extended to our original setting. 
The model in question is a different discretization to the one we have considered thus far: the so-called ``order discretization''.
In constructing the order discretization of $\mathcal{G}$, we strip away all information about the precise timing of edge labels and instead only retain the relative order between them.
The order discretization of $\mathcal{G}$ contains only a single edge at any given time and the edge at time $k$ corresponds to the $k$th edge to appear in $\mathcal{G}$.
Formally, it is defined as follows.
\begin{restatable}{definition}{DefOrder}
    For a temporal graph $\mathcal{G}=(V,E,\lambda)$ where no two edges share the same label, we define $\mathcal{G}^O=(V,E,\lambda^O)$ the \emph{order discretization} of $\mathcal{G}$ such that $\lambda^O(e)$ corresponds to the set of relative positions of each label in $\lambda(e)$ in the overall set of labels. 
    Furthermore, for $0\leq a<b$ we define $\mathcal{G}^O_{[a,b]}=(V,E_{[a,b]},\lambda^O_{[a,b]})$ to be $\mathcal{G}^O$ with all labels outside of $[a,b]$ removed.\medskip
\end{restatable}
Our main strategy for obtaining bounds on the doubling time of any specific set is to bound the time required for reachability balls to cross a sequence of cuts in the underlying graph.
As a consequence, we will need to make heavy use of the following bound on the proportion of labels from any cut of the underlying graph sampled in the order discretization before a given time. 
While in the general case, it is not difficult to show that any fixed cut, determined before any sampling occurs, has this property, we need to select the cuts for consideration based on the growth of the ball so far. This introduces a great deal of dependence which we circumvent via the following result.
\begin{restatable}{lemma}{LemNoSmallCuts}
    \label{lem: no small cuts}
    For any constant $\alpha,\beta>0$,
    \[\Prob\left[\max_{S \subseteq V}\frac{\sum_{e \in S\times V\setminus S}\lambda^O_{[0,\alpha n \log{n}]}(e)}{(r|S|(n-|S|))}>\beta\right]=o(n^{-1}).\]
\end{restatable}
\subsection{Upper Bounds}
We will begin with our argument for obtaining an upper bound on the doubling time. However, before we begin we will need two technical results. A key property of our model is that it is in some sense ``time reversible'' as the probability of any set of edge labels appearing in a given order is the same as that of them falling in the reverse (or in fact any) order.
This provides us with a very useful tool as we are able to consider our normal temporal graph and variants of it where we have partially reversed its labeling interchangeably.
The next definition and observation (see also \cref{def: reach}) exploits this property to reduce bounding the doubling time of small sets of vertices to bounding the doubling time of large sets of vertices in the partially reversed graph.
\begin{restatable}{definition}{DefReverse}
    For any temporal graph $\mathcal{G}=(V,E,\lambda)$, we define:
    \begin{itemize}
        \item $\textbf{Double}^{-}(\mathcal{G})=\max_{S \subset V:2\leq |S|\leq \frac{n}{4}}\left(\text{Double}_{\mathcal{G}}(S)\right)$ to be the doubling time over small sets.
        \item $\textbf{Double}^{+}(\mathcal{G})=\max_{S \subset V:\frac{n}{4}<|S|\leq \frac{n}{2}}\left(\text{Double}_{\mathcal{G}}(S)\right)$ to be the doubling time over large sets.
        \item $\mathcal{R}(\mathcal{G},t)=(V,E,\lambda_{\mathcal{R}(\mathcal{G},t)})$ to be the partially reversed temporal graph obtained by reversing the occurrence times of all labels before $t$, i.e. taking $\lambda_{\mathcal{R}(\mathcal{G},t)}(e)=\{R_t(l):l \in \lambda(e)\}$, where $R_t(l)=\begin{cases}t-l \text{ if }t>l\\ l \text{ otherwise}\end{cases}.$
    \end{itemize}
\end{restatable}
\begin{restatable}{lemma}{LemTimeWalk}
    \label{lem: Timewalk}
    For any temporal graph $\mathcal{G}=(V,E,\lambda)$, $\textbf{Double}^{-}(\mathcal{R}(\mathcal{G},\textbf{Double}^{+}(\mathcal{G}))\leq \textbf{Double}^{+}(\mathcal{G})$.\medskip
\end{restatable}
\begin{proof}
    For ease of notation we will denote $D=\textbf{Double}^{+}(\mathcal{G})$ and $\mathcal{H}=\mathcal{R}(\mathcal{G},\textbf{Double}^{+}(\mathcal{G}))$. 
    Now for the sake of contradiction, assume that there exists a set $S$ with $2\leq |S|\leq \frac{n}{4}$ such that $\text{Double}_\mathcal{H}(S)\geq D$.
    Consider that for any two vertices $u$ and $v$ if there exists a path from $u$ to $v$ in $\mathcal{G}$ arriving before time $D$ there must exist a path from $v$ to $u$ in $\mathcal{H}$ arriving by time $D$.
    Thus, we must have that $T$ the set of vertices outside of $S$ that can reach $S$ in $\mathcal{G}$ before time $D$ has size at most $|S|-1$.
    Furthermore, since $|S|\leq \frac{n}{4}$, we must have that $|S \cup T|\leq \frac{n}{2}$ and so we can take $U\subset V$ such that $|U|=\lceil \frac{n}{2}\rceil$ and $S \cup T\subseteq U$.
    Now by definition $V\setminus U$ has doubling time at most $D$ in $\mathcal{G}$.
    Since $V \setminus U$ has size $\lfloor\frac{n}{2}\rfloor$, this implies that at most one vertex cannot be reached from $V\setminus U$ via a time-respecting path before time $D$ in $\mathcal{G}$.
    So, in particular, there must be at least one vertex from $V\setminus U$ with a path to some member of $S$ in $\mathcal{G}$ arriving by time $D$.
    This contradicts the definition of $T$ and so no such set $S$ can exist.
    Therefore every set with size at least $2$ and at most $\frac{n}{4}$ must have a doubling time in $\mathcal{H}$ of at most $D$.
\end{proof}
While stated for arbitrary temporal graphs, given the equal probability of sampling both $\mathcal{G}$ and $\mathcal{R}(\mathcal{G},\textbf{Double}^{+}(\mathcal{G}))$ in our setting, this turns any upper bound on the doubling time of large sets into an upper bound on the doubling time of small sets.
Our second key preliminary result is an upper bound on the time taken for the reachability ball centered on any specific set of vertices to grow to a certain size.
We obtain the result by a carefully constructed coupling between the total time between expansions of the reachability ball and a sum of geometric random variables, but defer the somewhat laborious calculations.
Denote by $\mathbf{H}(a,b)=H_{n-a}+H_{b-1}-H_{a-1}-H_{n-b}$, where $H_a$ is the $a$th harmonic number with the convention that $H_0=0$.
\begin{restatable}{lemma}{LemGrowthUpper}
\label{lem: growth upper}
    For any $\alpha\geq0$, let $\mathcal{H}_\alpha=\mathcal{G}^O_{[\alpha n \log{n},\infty)}$, where $\mathcal{G}=(V,E,\lambda)$ is a sample from $\Gamma_{[n]}(\psi)$ with $\psi$ a degenerate distribution. Then for any $S \subseteq V$ chosen with knowledge of $\mathcal{G}^O_{[0,\alpha n \log{n}]}$, $|S|<k\leq n$,  $\beta>0$ and $0<\gamma<1$ small, we have the following.\\
    If $|S|(n-|S|)<(k-1)(n-k+1)$,
    \[\Prob{\left[\ball{\mathcal{H}_\alpha}{S}{\frac{(1+\beta) n \mathbf{H}(|S|,k)}{2(1-\gamma)}+\alpha n \log{n}}<k\right]}\leq \exp\left[-\frac{|S|(n-|S|)}{n}\mathbf{H}(|S|,k)(\beta-\log{(1+\beta)})\right]+o(n^{-1}).\]
    Otherwise,
    \[\Prob{\left[\ball{\mathcal{H}_\alpha}{S}{\frac{(1+\beta) n \mathbf{H}(|S|,k)}{2(1-\gamma)}+\alpha n\log{n}}<k\right]}\leq \exp\left[-\frac{(k-1)(n-k+1)}{n}\mathbf{H}(|S|,k)(\beta-\log{(1+\beta)})\right]+o(n^{-1}).\]
\end{restatable}
With these tools established, we may give our upper bound on the doubling time in $\mathcal{G}^O$. 
\begin{restatable}{lemma}{LemDoubleUpper}
  \label{lem: order double upper} Let $\mathcal{G}=(V,E,\lambda)$ be a sample from $\Gamma_{[n]}(\psi)$, for $\psi$ a degenerate distribution. Then for any $\beta>0$, we have that 
  $\textbf{Double}(\mathcal{G}^O)\leq (1+\beta) n \log{n}$.\medskip
\end{restatable}
The proof of this result consists largely of applying \cref{lem: growth upper}, along with a concentration inequality and \cref{lem: Timewalk} in order to deal with the case of small sets. More specifically, observe that once all vertices can be reached by at least $\frac{3n}{4}$ other vertices, every set of size at least $\frac{n}{4}$ can reach all vertices. Thus, once this has occurred the reachability ball of all large sets of vertices must have doubled in size. In order to bound this time, we first show via an application of \cref{lem: growth upper} and Markov's inequality that almost all vertices can be reached by $\frac{3n}{4}$ other vertices by roughly time $n\log{n}$. From a second application of \cref{lem: growth upper} we find that the few remaining vertices are reached from the already well connected vertices within only approximately a further $n\log{n}$ rounds. The more general result then follows from using \cref{lem: Timewalk} to extend the bound to cover all smaller sets of vertices.
\begin{proof}
    Denote by $T$ the first time step such that all vertices can be reached by at least $\frac{3n}{4}$ other vertices.
    For a fixed $\gamma,\epsilon>0$ small, denote by $X$ the number of vertices that are not reachable by $\frac{3n}{4}$ other vertices by time $T_{1}=\frac{(1+\epsilon) n \mathbf{H}(1,\frac{3n}{4}+1)}{2(1-\gamma)}$. It follows from \cref{lem: growth upper} and the linearity of expectation that,
    \[\Exp[X]\leq n\cdot\left(\exp\left[-\frac{n-1}{n}\mathbf{H}(1,\frac{3n+4}{4})(\epsilon-\log{(1+\epsilon)})\right]+o(n^{-1})\right).\]
    For n sufficiently large and $\epsilon$ sufficiently small,
    \[\leq n^{1-\epsilon+\log{(1+\epsilon)}}.\]
    Now applying Markov's inequality, we have that,
    \[\Prob{[X>n^{1-\epsilon+\log{1+\epsilon}}\log{n}]}\leq \frac{1}{\log{n}}.\]
    Thus with probability $1-\frac{1}{\log{n}}$ we have a set of at least $n-n^{1-\epsilon+\log{1+\epsilon}}\log{n}$ vertices that can be reached by at least $\frac{3n}{4}$ other vertices which we shall denote by $A$.
    Trivially, $T$ must occur at the very latest when every member of $V\setminus A$ is reachable by a member of $A$ via a path starting after $T_1$. Since $T_1=O(n\log{n})$, we know from a second application of \cref{lem: growth upper} that this occurs at the very latest by $T_1 +\frac{(1+\epsilon)n \mathbf{H}(n-n^{1-\epsilon+\log{1+\epsilon}}\log{n},n)}{2(1-\gamma)}$ with probability $1-o(1)$.
    Taking a union bound over the probability of failures we get that with probability $1-o(1)$,
    \[T<\frac{(1+\epsilon)n}{2(1-\gamma)}\left(\mathbf{H}(1,\frac{3n}{4}+1)+\mathbf{H}(n-n^{1-\epsilon+\log{1+\epsilon}}\log{n},n)\right)\]
    \[=\frac{(1+\epsilon)n}{2(1-\gamma)}\left(H_{n-1}+H_{\frac{3n}{4}}-H_{\frac{n}{4}-1}+H_{n^{1-\epsilon+\log{1+\epsilon}}\log{n}}+H_{n-1}-H_{n-n^{1-\epsilon+\log{1+\epsilon}}\log{n}}\right).\]
    For any $\alpha>0$ small, and all $n$ large enough, we can apply \cref{lem: Harmonics},
    \[\leq \frac{(1+\epsilon)n\log{n}}{2(1-\gamma)}\left(\left(4-\epsilon+\log{(1+\epsilon)}+\frac{\log\log{n}}{\log{n}}\right)(1+\alpha)-2(1-\alpha)\right)\]
    \[= \frac{(1+\epsilon)n\log{n}}{2(1-\gamma)}\left(2+2\alpha +(1+\alpha)\left(-\epsilon+\log{(1+\epsilon)}+\frac{\log\log{n}}{\log{n}}\right)\right).\]
    % \[\leq \frac{(1+\epsilon)n}{2(1-\gamma)}\left(2(1+\alpha)\log{n}-(1-\alpha^2)\log{n}+(1+\alpha)\log{n}-(1-\alpha)\log{n}+(1+\alpha)(1-\epsilon+\log{(1+\epsilon)})\log{n}\right)\]
    % \[=\frac{n\log{n}}{2(1-\gamma)}(1+\epsilon)\left[\left(4-\epsilon+\log{(1+\epsilon)}\right)(1+\alpha)-(2+\alpha)(1-\alpha)\right]\]
    For any $\beta>0$, taking $\alpha=\gamma=\epsilon>0$ sufficiently small, we find,
    \[\leq (1+\beta)n\log{n}\]
    Thus, we find that, with probability $1-o(1)$, all vertices can be reached by at least $\frac{3n}{4}$ other vertices by time $(1+\beta)n\log{n}$. 
    Therefore, we must have that by time $(1+\beta)n\log{n}$ all sets of size at least $\frac{n}{4}$ can reach all vertices.
    Thus we have that with probability $1-o(1)$, $\textbf{Double}^+(\mathcal{G})\leq (1+\beta)n\log{n}$ and by the time reversible nature of the distribution that $\textbf{Double}^+(\mathcal{R}(\mathcal{G},(1+\beta)n\log{n}))\leq (1+\beta)n\log{n}$ with the same probability. 
    In which case it follows from \cref{lem: Timewalk} that $\textbf{Double}^-(\mathcal{G})\leq (1+\beta)n\log{n}$.
    Thus taking a union bound over the failure probabilities, we find that $\textbf{Double}(\mathcal{G})\leq (1+\beta)n\log{n}$
\end{proof}
\subsection{Lower Bounds}
    In this section we present a lower bound to match the upper bound of \cref{lem: order double upper}.
    However, this time, our lives will be significantly easier, as it suffices to consider only the doubling time of sets containing precisely $\frac{n}{2}$ vertices.
    As before, we need a bound on the doubling time of any specific set, and so make use of the following result analogous to \cref{lem: growth upper}, the proof of which we also defer.
    \begin{restatable}{lemma}{LemGrowthLower}
        \label{lem: growth lower}
        Let $\mathcal{G}=(V,E,\lambda)$ be a sample from $\Gamma_{[n]}(\psi)$, for $\psi$ a degenerate distribution, and let $S\subseteq V$ be a set of vertices such that $|S|=\frac{n}{2}$. Then for any $\alpha, \beta>0$ sufficiently small we have that,
        \[\Prob\left[\text{Double}_{\mathcal{G}^O}(S)<\frac{(1-\alpha)n\log{n}}{2(1+\beta)}\right]=O(n^{\frac{\log{2}-1}{5}})+O(n^{2(1-\beta)(\alpha+\log{(1-\alpha)})}).\]
    \end{restatable}
    With that established we may now give the main result of the section, a lower bound on the doubling time of $\mathcal{G}^O$.
    \begin{restatable}{lemma}{LemOrderDoubleLower}
        \label{lem: order double lower}
        Let $\mathcal{G}=(V,E,\lambda)$ be a sample from $\Gamma_{[n]}(\psi)$, with $\psi$ a degenerate distribution, then
        $\textbf{Double}(\mathcal{G}^O)\geq (1-\epsilon)n\log{n}$.\medskip
    \end{restatable}
    The key notion for the proof of this result is the idea of a ``blocker vertex'', i.e. a vertex who is not yet reachable by any member of some large set of vertices. By applying \cref{lem: growth lower} and Markov's inequality, followed by a simple counting argument, we are able to show that at least a constant fraction of vertices must be blockers around time $\frac{n\log{n}}{2}$.
    Then, via drift analysis, we are able to show that at least one blocker vertex will not be included in an edge before time $(1-\epsilon)n\log{n}$.
    Given the fact that the doubling time is lower bounded by the time when the last vertex stops being a blocker, and that a vertex must be included in an edge in order to stop being a blocker, we immediately obtain the bound.
    \begin{proof}
    Let $X$ be the number of sets of vertices of size $\frac{n}{2}$ which have a doubling time of at most $\frac{(1-\alpha)n\log{n}}{2(1+\alpha)}$ for some $\alpha>0$ small. We know from \cref{lem: growth lower} (where the $O(n^{\frac{\log{2}-1}{5}})$ vanishes into the term depending on $\alpha$) and the linearity of expectation that for some constant $c>0$,
    \[\Exp[X]\leq cn^{2(1-\alpha)(\alpha+\log{(1-\alpha)})}\cdot \binom{n}{\frac{n}{2}}.\]
    Therefore by Markov's inequality it holds that,
    \[\Prob\left[X\geq cn^{2(1-\alpha)(\alpha+\log{(1-\alpha)})}\log{n}\cdot \binom{n}{\frac{n}{2}}\right]\leq \frac{1}{\log{n}}.\]
    Therefore with probability $1-\frac{1}{\log{n}}$ there are at least $(1-cn^{(1-\alpha)(\alpha+\log{(1-\alpha}))})\binom{n}{\frac{n}{2}}$ sets with a doubling time of at least $\frac{(1-\alpha)n\log{n}}{2(1+\alpha)}$.
    
    Now, a set of size $\frac{n}{2}$ has not yet doubled its reachability ball, if and only if, there is some vertex outside of the set that it cannot yet reach.
    We call such an unreachable vertex a ``blocker'' vertex for that set.
    Furthermore, any individual blocker vertex can block at most $\binom{n-1}{\frac{n}{2}}=\frac{1}{n}\binom{n}{\frac{n}{2}}$ sets (precisely those that do not contain it).
    Since we know that with high probability there are $(1-n^{-\beta})\binom{n}{\frac{n}{2}}$ sets which are blocked for some $\beta>0$, for any constant $0<d<1$ it holds that there must be at least $dn$ blocker vertices.

    In order for set to have doubled its reachability ball, it must have no more blocker vertices and a blocker vertex can only stop being a blocker when it is included in an edge.
    What we will now show is that, with high probability, at least one blocker vertex is not included in an edge between $\frac{(1-\alpha)}{2(1+\alpha)}n\log{n}$ and $(1-\epsilon)n\log{n}$.
    \begin{claim}
        There exists at least one blocker vertex that does not receive an edge between time $\frac{(1-\alpha)}{2(1+\alpha)}n\log{n}$ and $(1-\epsilon)n\log{n}$ with probability $1-o(1)$.
    \end{claim}
    \begin{proof}
    Let $B$ be a set of $dn$ vertices chosen arbitrarily at time $\frac{(1-\alpha)}{2(1+\alpha)}n\log{n}$ to either be a subset or super-set of the set of the currently blocking vertices.
    Let $Y_t$ be the set of vertices from $B$ that have not yet received an edge between time $\frac{(1-\alpha)}{2(1+\alpha)}n\log{n}$ and $\frac{(1-\alpha)}{2(1+\alpha)}n\log{n}+t$, and $\mathcal{F}_t$ be its corresponding filtration.
    We have the following two bounds, for $t<2n\log{n}$ and any fixed $\gamma>0$ immediately,
    \[\Prob[Y_{t}-Y_{t+1}=1|\mathcal{F}_t, Y_t=y\geq 1]\leq \frac{2y(1+\gamma)(n-y)}{n(n-1)},\]
    \[\Prob[Y_{t}-Y_{t+1}=2|\mathcal{F}_t, Y_t=y\geq 2]\leq \frac{y(1+\gamma)(y-1)}{n(n-1)}.\]
    We define $T_{\text{connect}}=\inf\{t\geq 0:Y_t\leq 2\}$, take $\zeta=1$ and $g(y)=\log{y+1}$, then
    \begin{align}\Exp\left[e^{\zeta(g(Y_t)-g(Y_{t+1}))}|\mathcal{F}_t, Y_t=y\geq3\right]&=\Prob{[Y_t-Y_{t+1}=0|\mathcal{F}_t,Y_t=y\geq3]}\\&+\Prob{[Y_t-Y_{t+1}=1|\mathcal{F}_t,Y_t=y\geq3]}\frac{y+1}{y}\\&+\Prob{[Y_t-Y_{t+1}=2|\mathcal{F}_t,Y_t=y\geq3]}\frac{y+1}{y-1}.\end{align}
    Since the probabilities must sum to one, and the coefficient of the term corresponding to the event with no change is minimal,
    we obtain for all $t<n \log{n}$,
    \[\leq 1-\frac{2y(1+\gamma)(n-y)}{n(n-1)}-\frac{y(1+\gamma)(y-1)}{n(n-1)}+\frac{2(y+1)(1+\gamma)(n-y)}{n(n-1)}+\frac{(1+\gamma)(y+1)y}{n(n-1)}\]
    \[\leq 1+\frac{2(1+2\gamma)}{n}.\]
    We now apply \cref{lem: Drift},
    to obtain that, for any $\eta>0$ small,
    \[\Prob\left[T_{\text{connect}}<\frac{(1-\eta)n\log{n}}{2}\right]\leq \left(1+\frac{2(1+2\gamma)}{n}\right)^{\frac{(1-\eta)n\log{n}}{2}}\cdot e^{{\log{3}-\log{dn+1}}}\]
    \[\leq \frac{3}{d}n^{(1-\eta)(1+2\gamma)-1}\]
    \[\leq \frac{3}{d}n^{2\gamma-\eta-2\eta \gamma}.\]
    Taking $\gamma<\frac{\eta}{2}$ gives us that with probability $1-o(1)$ there is at least one member of $B$, that has not been included in an edge after $\frac{(1-\alpha)}{2(1+\alpha)}n\log{n}$ and before $\frac{(1-\alpha)}{2(1+\alpha)}n\log{n}+\frac{1-\eta}{2}n\log{n}$.
    However, we also know that with probability $1-o(1)$ there are at least $dn$ blocker vertices.
    Therefore, by construction, with probability $1-o(1)$, $B$ only contains blocker vertices.
    A union bound over the probability of failure gives the claim.
    \end{proof}
    The main claim follows immediately from this, the consolidation of constant terms and the observation that if there is at least one blocker vertex that has not received an edge, there must be at least one set that has not yet doubled its size.
    \end{proof}
    \subsection{Tidying Up}
    In this section, we prove our main theorem on the doubling time in the continuous model.
    In order to achieve this we will convert our upper and lower bounds on the order discretization of a sample from the continuous model back to the undiscretized version.
    This is made quite easy by the strong concentration of the number of labels occurring within a given interval. 
    In fact, since $\psi$ is degenerate we obtain the ever pleasant binomial distribution.
    We are able, therefore, to obtain the following simple precursor, showing that with high probability the number of labels appearing in the initial interval is highly concentrated.
    \begin{restatable}{lemma}{LemConnectivityBridge}
        \label{lem: connectivity bridge}
        For $\alpha, \beta>0$ with $\beta$ small, the following holds, for any $\mathcal{G}=(V,E,\lambda)$,
        \[\Prob\left[\sum_{e \in E}|\lambda^O_{[0,\frac{\alpha \log{n}}{r(n-1)}]}(e)|\leq \frac{\alpha(1-\beta)n\log{n}}{2}\right]=o(n^{-1}).\]
        Similarly,
        \[\Prob\left[\sum_{e \in E}|\lambda^O_{[0,\frac{\alpha \log{n}}{r(n-1)}]}(e)|\geq \frac{\alpha(1+\beta)n\log{n}}{2}\right]=o(n^{-1}).\]
    \end{restatable}
    \begin{proof}
        Let $X$ be the number of labels falling before $\frac{\alpha\log{n}}{r(n-1)}$. Since each label falls there independently with probability $\frac{\alpha\log{n}}{r(n-1)}$, we find that,
        \[\Exp[X]=\frac{\alpha n \log{n}}{2} .\]
        Thus from a Chernoff bound (\cref{Chernoff Bound}) we find that,
        \[\Prob{[X<\frac{(1-\beta) \alpha n \log{n}}{2}]}=n^{-\frac{\alpha\beta^2n}{4}}.\]
        Similarly,
        \[\Prob{[X>\frac{\alpha(1+\beta)n\log{n}}{2}]}\leq n^{-\frac{\alpha\beta^2n}{6}}.\]
        In both cases these are $o(n^{-1})$.
    \end{proof}
    The proof of the result itself is then simply an application of \cref{lem: connectivity bridge} to the bounds implied by \cref{lem: order double lower} and \cref{lem: order double upper}.
    \begin{proof}[Proof of \cref{thm: continuous double}]
        It follows from \cref{lem: order double upper}, \cref{lem: order double lower} and a union bound that for any $\epsilon>0$,
        \[\Prob[(1-\epsilon)n\log{n}<\textbf{Double}(\mathcal{G}^O)<(1+\epsilon)n\log{n}]=1-o(1).\]
        Furthermore, we know from \cref{lem: connectivity bridge} that for any $\delta>0$ the $(1-\epsilon)n\log{n}$th label falls after $\frac{2(1-\delta)(1-\epsilon)\log{n}}{r(n-1)}$ and that the $(1+\epsilon)n\log{n}$th labels falls before $\frac{2(1+\delta)(1+\epsilon)\log{n}}{r(n-1)}$ with probability $1-o(1)$.
        Thus taking a union bound over the failure probabilities, and $\alpha>\delta+\epsilon+\delta\epsilon$, we obtain the claim.
    \end{proof}
\section{Auxiliary Results}
At various points in the paper we make use of the following series of well-known results.
\begin{lemma}[Folklore bound on the Harmonic Series]
    \label{lem: Harmonics}
    For any $\alpha>0$, for all $n$ large enough,
    \[(1-\alpha)\log{n}\leq \sum_{i=1}^n\frac{1}{i}\leq (1+\alpha)\log{n}.\]
\end{lemma}
\begin{lemma}[Markov's Inequality (See for example Lemma 27.1. of \cite{frieze2015introduction})]
    \label{Markov's Inequality}
    Let $X$ be a non-negative random variable. Then for any $\delta>1$,
    \[\Prob{[X\geq \delta \Exp[X]]}\leq \frac{1}{\delta}.\]
\end{lemma}
\begin{lemma}[Chernoff Bound (Adapted from Theorems 4.4 and 4.5 of \cite{mitzenmacher2017probability})]
    \label{Chernoff Bound}
    Let $\mathbf{X}=\sum_{i=0}^nX_i$, where $X_1,...,X_n$ are independent identically distributed Bernoulli random variables. Then for any $\delta>0$,
    \[\Prob{[\mathbf{X}\geq (1+\delta)\Exp[\mathbf{X}]]}\leq e^{-\frac{\delta^2\Exp[\mathbf{X}]}{3}}.\]
    Further, for any $0<\delta<1$,
    \[\Prob{[\mathbf{X}\leq (1-\delta)\Exp[\mathbf{X}]]}\leq e^{-\frac{\delta^2\Exp[\mathbf{X}]}{2}}.\]
\end{lemma}
\begin{lemma}[Tail Bound on the sums of Geometric Random Variables (Adapted from Theorems 2.1 and 3.1 of \cite{JANSON20181})]
    \label{lem: Janson}
    Let $\mathbf{X}=\sum_{i=0}^k X_i$ be the sum of the independent geometric random variables $X_i\sim \Geom(p_i)$ where $0<p_i\leq 1$. Then for any $0<\alpha\leq 1$,
    \[\Prob{\left[\mathbf{X}\leq \alpha \Exp\left[\mathbf{X}\right]\right]}\leq \exp\left[-(\alpha-1-\log{\alpha})\left(\min_{0\leq i \leq k}(p_i)\right)\Exp{[\mathbf{X}]}\right].\]
    Similarly, for any $\beta>1$,
    \[\Prob{\left[\mathbf{X}\geq \beta \Exp\left[\mathbf{X}\right]\right]}\leq \exp\left[-(\beta-1-\log{\beta})\left(\min_{0\leq i \leq k}(p_i)\right)\Exp{[\mathbf{X}]}\right].\]
\end{lemma}
\begin{lemma}[Drift Theorem (Adapted from Theorem 2 of \cite{DriftTheorem})]
    \label{lem: Drift}
    Let $(X_t)_{t\in \mathbb{N}_0}$, be a stochastic process, adapted to a filtration $(\mathcal{F}_t)_{t\in \mathbb{N}_0}$, over some state space $S \subseteq \mathbb{R}_\geq 0$. For some $a\geq 0$, such that $\{x \in S|x\leq a\}$ is absorbing, let $T_a=\min\{t\geq 0:X_t\leq a\}$. Moreover, let $g:S\to \mathbb{R}_{\geq 0}$ be a function such that $g(0)=0$ and $g$ is increasing on $\mathbb{R}_{\geq 0}$. If there exists $\lambda>0$ and a value $\beta>0$ such that the following holds for all $t\in \mathbb{N}_{0}$,
    \[\Exp\left[e^{\lambda(g(X_t)-g(X_{t+1}))}\mathbbm{1}_{X_t>a}-\beta\middle|\mathcal{F}_{t}\right]\leq 0.\]
    Then for $t^*>0$ and $X_0>a$,
    \[\Prob{[T_{a}>t^*|\mathcal{F}_0]}<\beta^{t^*}e^{\lambda g(a)-\lambda g(X_0)}.\]
\end{lemma}
\begin{lemma}
    \label{lem: dependent Bernoulli}
    Let $X$ be a random variable with a continuous cumulative distribution function supported entirely on $[0,1]$, such that $E[X|X>0]$ and $E[X^2|X>0]$ both exist. Further let $Y,Z \sim\text{Ber}(X)$ be independent samples from a Bernoulli random variable parameterized by the same sample from $X$, then \[\Prob{[Y=1|Z=1]}\geq \Prob{[Y=1|X>0]}\]
\end{lemma}
\begin{proof}
    We have the following,
    \[\Prob{[YZ=1|X>0]}=\Exp[X^2|X>0]\geq \Exp[X|X>0]^2=\Prob{[Y=1|X>0]}\Prob{[Z=1|X>0]}\]
    Thus dividing both sides by $\Prob{[Z=1|X>0]}$,
    \[\Prob{[Y=1|Z=1]}=\Prob{[Y=1|Z=1,X>0]}\geq \Prob{[Y=1|X>0]}\]
\end{proof}
\section{Deferred Proofs}
In this section we provide the less interesting proofs which have been excluded from the main body of the text.
\subsection{ Motifs}
We begin with the deferred proofs for the motifs.
\subsubsection{First Moment}
    \begin{proof}[Remaining proof for \cref{lem: expectation}]
    We first consider the case where $\beta>1$.
    By linearity, we obtain,
    \[\Exp[\DiscSum{\beta}{T}]=\sum_{S \in \mathcal{S}}\sum_{\tau_S \in \mathcal{T}_S}\Exp[\LabEx_{S,\tau}\DiscDel{\beta}{T}_{S,\tau}]\]
    By independence we have that,
    \[=\sum_{S \in \mathcal{S}}\sum_{\tau \in \mathcal{T}_S}\Exp[\LabEx_{S,\tau}]\Exp[\DiscDel{\beta}{T}_{S,\tau}]\]
    From \cref{lem: bridge 3} we have,
    \[\geq\left(\frac{\beta}{2|E_H|T}\right)^{|E_H|}\left(1+\frac{|E_H|(1-\frac{\beta}{2T})}{\frac{\beta}{2T}}\right)\sum_{S \in \mathcal{S}}\sum_{\tau \in \mathcal{T}_S}\Exp[\LabEx_{S,\tau}]\]
    By independence,
    \[=\left(\frac{\beta}{2|E_H|T}\right)^{|E_H|}\left(1+\frac{|E_H|(1-\frac{\beta}{2T})}{\frac{\beta}{2T}}\right)\sum_{S \in \mathcal{S}}\sum_{\tau \in \mathcal{T}_S}\prod_{e \in E_S}\Prob{\left[\tau(e) \leq l_e\right]}\]
    Since $\sum_{\tau \in \mathcal{T}_S}\Prob{[\tau(e)\leq l_e]}=\Exp[\psi]=r$ is both finite and absolutely convergent, by the repeated application of distributivity over absolutely convergent sequences, we obtain,
    \[=\left(\frac{\beta}{2|E_H|T}\right)^{|E_H|}\left(1+\frac{|E_H|(1-\frac{\beta}{2T})}{\frac{\beta}{2T}}\right)\sum_{S \in \mathcal{S}}\prod_{e \in E_S} r\]
    \[=\left(\frac{\beta r}{2|E_H|T}\right)^{|E_H|}\left(1+\frac{|E_H|(1-\frac{\beta}{2T})}{\frac{\beta}{2T}}\right)|V_H|!\binom{n}{|V_H|}\]
    Now for the case $\beta=1$, from the above we have that
    \[\Exp[\DiscSum{1}{T}]=\sum_{S \in \mathcal{S}}\sum_{\tau \in \mathcal{T}_S}\Exp[\LabEx_{S,\tau}]\Exp[\DiscDel{1}{T}_{S,\tau}]\]
    From \cref{lem: bridge 3} we have,
    \[=T^{1-|E_H|}\sum_{S \in \mathcal{S}}\sum_{\tau \in \mathcal{T}_S}\Exp[\LabEx_{S,\tau}]\]
    \[=T^{1-|E_H|}\sum_{S \in \mathcal{S}}\sum_{\tau \in \mathcal{T}_S}\prod_{e \in E_S}\Prob{\left[\tau(e) \leq l_e\right]}\]
    Since $\sum_{\tau \in \mathcal{T}_S}\Prob{[\tau(e)\leq l_e]}=\Exp[\psi]=r$ is both finite and absolutely convergent, by the repeated application of distributivity over absolutely convergent sequences,
    \[=T^{1-|E_H|}\sum_{S \in \mathcal{S}}\prod_{e \in E_S} r\]
    \[=T^{1-|E_H|}\binom{n}{|V_H|}|V_H|!\]
    \end{proof}
    
    \subsubsection{Second Moment}
    \begin{proof}[Proof of \cref{corr: Weak-disc-upper} in the case where $\beta>1$]
    Since $\Prob{[\DiscSum{\beta}{T}=0]}\leq \frac{\Var[\DiscSum{\beta}{T}]}{\Exp[\DiscSum{\beta}{T}]^2}$, it follows from \cref{lem: variance} and \cref{lem: expectation}
    \[\Prob{[\DiscSum{\beta}{T}=0]}\leq \frac{(\frac{\sqrt{2}\beta r}{T} )^{2|E_H|}n^{2|V_H|}(2|V_H|)^{4|V_H|}\left(1+\frac{|E_H|(T-\beta)}{\beta}\right)\sum_{I \subseteq H}n^{-|V_I|}
    \left(\frac{2\beta r}{T}\right)^{-|E_I|}}{\left(\left(\frac{\beta r}{2T|E_H| }\right)^{|E_H|}|V_H|! \binom{n}{|V_H|}\left(1+\frac{|E_H|(2T-\beta)}{\beta}\right)\right)^2}\]
    Using $\binom{a}{b}\geq \left(\frac{a}{b}\right)^b$ and $x!\leq x^x$
    \[=O(1)\frac{\sum_{I \subseteq H} n^{-|V_I|}\Theta(\frac{T}{\beta})^{|E_I|}}{\Theta{\left(\frac{T}{\beta}\right)}}\]
    \[=\Theta(\max_{I \subseteq H}n^{-|V_I|}\left(\frac{\beta}{T}\right)^{1-|E_I|})=o(1)+\Theta(\max_{I \sqsubseteq H}n^{-|V_I|}\left(\frac{\beta}{T}\right)^{1-|E_I|})\]
    Since we have that $\frac{\beta}{T}=\omega\left(n^{-\min_{I \sqsubseteq H}\frac{|V_I|}{|E_I|-1}}\right)$
    \[=o(1)\]
    \end{proof}
    \begin{proof}[Proof of \cref{corr: Weak-disc-upper} in the case where $\beta=1$]
        Since $\Prob{[\DiscSum{\beta}{T}=0]}\leq \frac{\Var[\DiscSum{\beta}{T}]}{\Exp[\DiscSum{\beta}{T}]^2}$, it follows from \cref{lem: variance} and \cref{lem: expectation}
        \[\Prob{[\DiscSum{1}{T}=0]}\leq \frac{(\frac{\sqrt{2} r}{T} )^{2|E_H|}n^{2|V_H|}(2|V_H|)^{4|V_H|}\left(1+|E_H|(T-1)\right)\sum_{I \subseteq H}n^{-|V_I|}
        \left(\frac{2 r}{T}\right)^{-|E_I|}}
        {\left(r^{|E_H|}T^{1-|E_H|}|V_H|! \binom{n}{|V_H|}\right)^2}\]
        Using $\binom{a}{b}\geq \left(\frac{a}{b}\right)^b$ and $x!\leq x^x$ this simplifies to,
        \[=O(1)\cdot\frac{1+|E_H|(T-1)}{T^{2}}\sum_{I \subseteq H}n^{-|V_I|}
        \left(\frac{2r}{T}\right)^{-|E_I|}\]
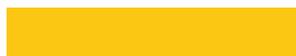
            \[=O(1)\cdot\sum_{I \subseteq H}n^{-|V_I|}
        \Theta(T^{|E_I|-1})\]
        \[=o(1)+O(1)\cdot\sum_{I \sqsubseteq H}n^{-|V_I|}
        \Theta(T^{|E_I|-1})\]
        Since we have that $\frac{1}{T(n)}=\omega\left(n^{-\min_{I \sqsubseteq H}\frac{|V_I|}{|E_I|-1}}\right)$
        \[=o(1)\]
    \end{proof}
    \subsection{Doubling Times}
    \begin{proof}[Proof of \cref{lem: no small cuts}]
        Fix a set $S \subseteq V$, such that $s=|S|\leq\frac{n}{2}$, and consider the number of labels sampled from $\text{Cut}=S\times V\setminus S$ by time $\alpha n \log{n}$, which we denote by $X_S$. On round $t$ the probability of sampling an edge from the cut is given by the following expression
        \[p=\frac{2}{rn(n-1)-2t} \cdot \left(rs(n-s)-\sum_{e \in \text{Cut}}\lambda^O_{[0,t)}(e)\right)\]
        For $n$ sufficiently large, any fixed $S$, $t\leq \alpha n \log{n}$ and $\gamma>0$, this is trivially upper bounded by,
        \[\leq \frac{2(1+\gamma)s(n-s)}{n(n-1)}\]
        Therefore, we have that $X_S$ is stochastically dominated by $Y_S\sim \Bin\left(\alpha n \log{n}, \frac{2(1+\gamma)s(n-s)}{n(n-1)}\right)$, and thus $\Exp[X_S]\leq \frac{2\alpha s(1+\gamma)(n-s)\log{n}}{n-1}$.
        Applying a Chernoff bound (\cref{Chernoff Bound}), we have that for any $k>0$,
        \[\Prob{\left[X_S\geq (1+k)\frac{2\alpha s(1+\gamma)(n-s)\log{n}}{n-1}\right]}\leq\Prob{[Y_S\geq (1+k)\Exp[Y]]}\]
        \[\leq \exp\left[-\frac{k^2}{3}\cdot \frac{2\alpha s(1+\gamma)(n-s)\log{n}}{n-1}\right]\]
        Since $s\leq \frac{n}{2}$,
        \[\leq \exp\left[-\frac{k^2}{9}\cdot 2\alpha s(1+\gamma)\log{n}\right]\]
        Now for any $s\leq \frac{n}{2}$ via a union bound,
        \[\Prob\left[\bigcup_{S \subseteq V:|S|=s}\left\{ X_S\geq (1+k)\frac{2\alpha s(1+\gamma)(n-s)\log{n}}{n-1}\right\}\right]\leq |\{S \subseteq V: |S|=s\}|\exp\left[-\frac{k^2}{9}\cdot 2\alpha s(1+\gamma)\log{n}\right]\]
        \[\leq \exp\left[s\log{n}\left(1-\frac{2\alpha k^2(1+\gamma)}{9}\right)\right]\]
        This is clearly maximized when $s=1$, however even in this case we can take $k$ to be a sufficiently large constant that,
        \[=o(n^{-2})\]
        Then via a second union bound over the failure probability for each $s \in [\frac{n}{2}]$ we obtain that,
        \[\Prob\left[\bigcup_{S \subseteq V} \left\{X_S\geq (1+k)\frac{2\alpha |S|(1+\gamma)(n-|S|)\log{n}}{n-1}\right\}\right]=o(n^{-1})\]
        The result then follows from a simple rearrangement, the fact that $r\geq 1$ and that \[(1+k)\frac{2\alpha |S|(1+\gamma)(n-|S|)\log{n}}{n-1}=o(|S|(n-|S|))\]
    \end{proof}
    \begin{proof}[Proof of \cref{lem: growth upper}]
    Our strategy shall be to bound the time between vertices being added to the reachability ball via suitably nice geometric random variables. 
    We will begin with the geometric random variables.
    For some fixed $\gamma>0$ and all $i \in [n-1]$, we define $X_i\sim \Geom\left(\frac{2i(1-\gamma)(n-i)}{n(n-1)}\right)$. Let $\mathbf{X}=\sum_{i=1}^{n-1}X_i$ and $\mathbf{\bar{X}}=\sum_{i=s}^{k-1}X_i$, where $s=|S|$.
    \begin{claim}
        \label{clm: well behaved geo}
        \[\Prob[\mathbf{X}\geq 10n\log{n}]=o(n^{-6})\]
        Furthermore for any $\delta>0$,
        \[\Prob{\left[\mathbf{\bar{X}}\geq (1+\delta)\frac{(n-1)}{2(1-\gamma)}\mathbf{H}\left(s,k\right)\right]}\leq \exp\left[\frac{\min((k-1)(n-k+1),s(n-s))}{n}\mathbf{H}(s,k)(\delta-\log{(1+\delta)})\right]\]
    \end{claim}
    \begin{proof}
        First considering $\mathbf{X}$,
        \[\Exp[\mathbf{X}]=\frac{n(n-1)}{2(1-\gamma)}\sum_{i=1}^{n-1}\frac{1}{i(n-i)}=\frac{(n-1)}{2(1-\gamma)}\sum_{i=1}^{n-1}\frac{1}{i}+\frac{1}{n-i}\]
        \[=\frac{n-1}{1-\gamma}\sum_{i=1}^{n-1}i^{-1}\]
        By \cref{lem: Harmonics} and taking $n$ sufficiently large,
        \[n\log{n}\leq \Exp[\mathbf{X}]\leq \frac{n\log{n}}{(1-\gamma)^2}\]
        Applying the upper tail bound of \cref{lem: Janson} for $\delta>0$ we have that,
        \[\Prob[\mathbf{X}\geq (1+\delta)\frac{n\log{n}}{1-\gamma}]\leq \exp\left[-\left(\frac{(1-\gamma)}{n}\left(n\log{n}\right)\left(\delta-\log{(1+\delta)}\right)\right)\right]\]
        We will take $\delta=10(1-\gamma)-1$ and find that $\Prob[\mathbf{X}\geq 10n\log{n}]=o(n^{-6})$.
    
        Now, taking $s=|S|$, consider $\mathbf{\bar{X}}=\sum_{i=s}^{k-1}X_i$.
        By a similar approach to above we have that,
        \[\Exp[\mathbf{\bar{X}}]=\frac{(n-1)}{2(1-\gamma)}\sum_{i=s}^{k-1}\frac{1}{i}+\frac{1}{n-i}=\frac{(n-1)}{2(1-\gamma)}\left(H_{n-s}+H_{k-1}-H_{s-1}-H_{n-k}\right)\]
        Applying \cref{lem: Janson} again, we find that for any $\delta>0$ that,
        \[\Prob{[\mathbf{\bar{X}}\geq (1+\delta)\Exp[\mathbf{\bar{X}}]]}\leq \exp\left[-\left(\min_{i \in \{s,k-1\}}(\frac{2i(n-i)(1-\gamma)}{n(n-1)})\right)\left(\delta-\log{(1+\delta)}\right)\Exp[\mathbf{\bar{X}}]\right]\]
        In the case that $s(n-s)>(k-1)(n-k+1)$, this simplifies to,
        \[=\exp\left[-\frac{s(n-s)}{n}\mathbf{H}(s,k)(\delta-\log{(1+\delta)})\right]\]
        In the other case we have that,
        \[=\exp\left[-\frac{(k-1)(n-k+1)}{n}\mathbf{H}(s,k)(\delta-\log{(1+\delta)})\right]\]
        Which gives the claim.\\
    \end{proof}
    With that established, we can start working with the random temporal graph itself.
    We define $A_t=|\ball{\mathcal{G}^O_{[\alpha n \log{n}, \infty)}}{S}{\alpha n \log{n}+t}|$ to be the size of the reachability ball centered at $S$ after $t$ rounds.
    Further, we define $B_i=\inf\{t\geq 0: A_t\geq s+i\}$.
    We will take a slightly  different (but equivalent) view of the process sampling $\mathcal{G}^O$: Each edge begins with $r$ tokens. 
    Then on each round, an edge is chosen to become active and discard a token with probability proportional to its remaining tokens, the process concluding once all tokens have been sampled.

    % We use $\mathbf{C}_i$ to denote the history of this process up to and including round $\alpha n \log{n}+B_i$, and then for $C$ any valid history the following holds for each $t>B_i$ such that $\Prob[B_{i+1}=t+1|B_{i+1}>t, \mathbf{C}_i=C]< 1$
    % \[\Prob[B_{i+1}=t+1|B_{i+1}>t, \mathbf{C}_i=C]\leq \Prob[B_{i+1}=t+2|B_{i+1}>t+1, \mathbf{C}_i=C]\]
    % Thus, $B_{i+1}-B_{i}$ corresponds to the number of independent random trials up to and including the first success, where the probability of success in each successive trial is non-decreasing.
    % Therefore, $B_{i+1}-B_i$ is stochastically dominated by $Y_i$ the geometric random variable with a success probability equal to that of the first trial.

    % Unfortunately, the success parameter of $Y_i$ is itself a non-trivial random variable and the parameter of $Y_{i}$ depends on the history up to $B_i$.
    % What we shall show, however, is that $\sum_{i=0}^{k-s-1}Y_i\leq \mathbf{\bar{X}}$ with high probability under an appropriate coupling. 
    % However, in order to achieve this we shall have to approach this iteratively.
    We prove the following claim as both a base case and a warm up.
    \begin{claim}
        Under an appropriate coupling \[\Prob[B_1<X_s]=1-o(n^{-1})\]
    \end{claim}
    \begin{proof}
    We construct the coupling as follows. If there exists a set $T$ such that $\sum_{e \in T \times V\setminus T}\frac{|\lambda^O_{[0,\alpha n \log{n}]}(e)|}{r|T|(n-|T|)}>\gamma$ we immediately terminate the coupling, declare it to have failed and sample $X_s$ independently of $\mathcal{G}$. If this is not the case, we begin on round $\alpha n \log{n}$ and work round by round. On the $t$th subsequent round,
    \begin{itemize}
        \item If no edge has been revealed crossing the cut. We reveal the next edge of $\mathcal{G}^O$.
        \begin{enumerate}
            \item If the edge revealed is from the cut, we sample the $t$th trial of $X_s$ to be a success with probability 
            \[q_t=\frac{n(n-1)-2\alpha n\log{n}-2t}{n(n-1)}\cdot\frac{(1-\gamma)rs(n-s)}{rs(n-s)-\sum_{e \in S\times V\setminus S}|\lambda^O_{[0,\alpha n \log{n}]}(e)|}\]
            \item Otherwise we sample a failure for $X_s$
        \end{enumerate}
        \item If an edge has been revealed crossing the cut, we sample the $t$th trial for $X_s$ to be a success independently at random according to its standard success probability.
    \end{itemize}
    By construction, we have that, unless the coupling fails, $X_s$ samples its first successful trial no earlier than an edge is revealed crossing the cut. However, the coupling only fails if \cref{lem: no small cuts} fails and so we have the claim.\\
    \end{proof}
    We will now extend this argument inductively, to obtain the following,
    \begin{claim}
        \label{clm: Upper coupling}
        Under an appropriate coupling, \[\Prob[B_{k-s}<\mathbf{\bar{X}}]= 1-o(n^{-1})\]
    \end{claim}
    \begin{proof}
    We shall chain together $k-s$ variations of the previous coupling and use the good behavior of each one to ensure the good behavior of the next. The coupling is divided into $k-s$ phases such that, for $i \in \{0,...,k-s-1\}$, the $i$th phase begins on round $B_i$. In the $i$th phase we couple $B_{i+1}-B_i$ and $X_{s+i}$, taking \[\text{Cross(i)}=\ball{\mathcal{G}^O_{[\alpha n \log{n},\infty)}}{S}{\alpha n \log{n}+B_i}\times (V \setminus \ball{\mathcal{G}^O_{[\alpha n \log{n},\infty)}}{S}{\alpha n \log{n}+B_i})\] to be the set of edges crossing the cut between the reached and unreached vertices.
    Firstly, if either of the following hold, we immediately terminate the coupling and sample $X_s+i$ independently at random:
    \begin{enumerate}
        \item There exists a set $T$ such that $\sum_{e \in T \times V\setminus T}\frac{|\lambda^O_{[0,\alpha n \log{n}+B_i]}(e)|}{r|T|(n-|T|)}>\gamma$ (i.e. the cut between $T$ and $V\setminus T$ has had too many labels sampled from it).
        \item $B_i>(\alpha+10)n\log{n}$ (i.e. the process has taken improbably long already).
    \end{enumerate}
    Then, if neither of these occur, we work round by round starting at $\alpha n\log{n}+B_i$ and on the $t$th subsequent round.
    \begin{itemize}
        \item If no edge has been revealed from $\text{Cross}(i)$ we reveal the next edge of $\mathcal{G}^O$.
        \begin{enumerate}
            \item If the edge revealed is from $\text{Cross}(i)$, we sample the $t$th trial for $X_{s+i}$ to be a success with probability
            \[q_t=\frac{n(n-1)-2\alpha n\log{n}-2t-2B_i}{n(n-1)}\cdot\frac{(1-\gamma)rs(n-s)}{rs(n-s)-\sum_{e \in S\times V\setminus S}|\lambda^O_{[0,\alpha n \log{n}+B_i]}(e)|}\]
            \item Otherwise, we sample a failure for the $t$th trial of $X_{s+i}$.
        \end{enumerate}
        \item If an edge has been revealed from $\text{Cross}(i)$, we sample the $t$th trial for $X_{s+i}$ according to its standard success probability, independently of $\mathcal{G}$.
    \end{itemize}
    Now note that unless the coupling fails for some $i \in \{0,...,k-s-1\}$ we have that $B_{i+1}-B_i\leq X_{s+i}$ for all of the same $i$, and thus $B_{k-s}\leq\mathbf{\bar{X}}$.
    It therefore remains to show that the coupling does not fail with high probability.
    
    There are two possibilities, either the coupling fails in phase $0$ or it fails in a later phase.
    Since we know from the previous lemma that the coupling succeeds in the zeroth phase with probability $1-o(1)$, any substantial failure probability must occur later.
    Accordingly, let $j \in \{1,...,k-s-1\}$ be the first phase in which the coupling fails.
    This requires that either, $B_j>10n\log{n}$ or that there is a cut that has been over-sampled by time $\alpha n\log{n}+B_j$.
    In turn, this requires that either $B_{j-1}+X_{s+j-1}>10n\log{n}$ (and therefore $\mathbf{X}>10n\log{n}$) or that there is a cut that has been over-sampled by time $(\alpha+10n\log{n})$.
    However, from \cref{clm: well behaved geo} and \cref{lem: no small cuts} respectively, we know that both occur with probability at most $o(n^{-1})$.
    Therefore, we have that in order for the coupling to fail, at least one of three events must occur, each of which does so with probability at most $o(n^{-1})$.
    It thus follows via a union bound over the failure probabilities, that the coupling succeeds with probability $1-o(n^{-1})$.
    As such, we have the claim.

    % We shall proceed inductively, assuming that $B_i<10n\log{n}$ and that no cut has been over sampled by time $\alpha n \log{n}+B_i)$.
    % From the previous claim and \cref{clm: well behaved geo}, we have that $B_1<{X_s}<10n\log{n}$ with probability $1-o(1)$.
    % We begin by observing that it holds that with probability $1-o(n^{-1})$, we have that $\mathbf{\bar{X}}<\mathbf{X}\leq 10n\log{n}$ (\cref{clm: well behaved geo}) and that no cut has had more than a $1-\gamma$ proportion of its labels used by $(\alpha+10)n \log{n}$ (\cref{lem: no small cuts}). From the previous claim, we know that the coupling does not fail in the $0$th phase with probability $1-o(n^{-1})$.
    % Now assume that the coupling first fails in phase $i$. This requires that either $B_i>10n\log{n}$ or that there is a sufficiently over taxed cut by time $\alpha n \log{n}+B_i$.
    % In the former case, since all previous phases did not fail we must have that $10n\log{n}\leq \sum_{j=s}^{s+i}X_j\leq \mathbf{X}$. In the latter case, either the former case must be true, or more than a $(1-\beta)$ fraction of edges must have been taken from some cut before $(\alpha+10)n\log{n}$ rounds.
    % Iterating this argument, we find that for any phase to fail, we require that either $\mathbf{X}\geq 10n\log{n}$ or that there exists a set $T$ such that $\sum_{e \in T \times V\setminus T}\frac{|\lambda^O_{[0,(\alpha+10) n \log{n}]}(e)|}{r|T|(n-|T|)}>\gamma$, the probability of which is at most $o(n^{-1})$ via \cref{clm: well behaved geo}, \cref{lem: no small cuts} and a union bound.

    Thus with probability $1-o(n^{-1})$ we find that $B_{k-s}<\mathbf{\bar{X}}$.
    \end{proof}
    A final union bound over the failure probabilities of \cref{clm: Upper coupling} and \cref{clm: well behaved geo} gives the result.
    \end{proof}

    \begin{proof}[Proof of \cref{lem: growth lower}]
        As in the equivalent upper bound, our strategy is to bound the time between vertices being added to the reachability ball via suitably nice geometric random variables.
        Similar to before, we define $A_t=|\ball{\mathcal{G^O}}{S}{t}|$ to be the size of the reachability ball centered at $S$ after $t$ rounds. 
        Further, we define $B_i=\inf\{t>0:A_t\geq \frac{n}{2}+i\}$.
        Our goal is to lower bound $B_{i+1}-B_i$ via the independent random variable $Y_i \sim \Geom\left(p_i\right)$ for $p_i=(1+\beta)\frac{2(\frac{n}{2}+i)(\frac{n}{2}-i)}{n(n-1)}$ 
        \begin{claim}
            \label{clm: lower coupling}
            $\Prob[B_{\frac{n-2}{2}}>\sum_{i=0}^{\frac{n-2}{2}}Y_i]=1-o(n^{\frac{\log{2}-1}{5}})$
        \end{claim}
        \begin{proof}
        We construct the following coupling between $B_{i+1}-B_i$ and $Y_i$, where we denote by $\text{Cross}(i)=\ball{\mathcal{G}^O}{S}{B_i}\times (V \setminus \ball{\mathcal{G}^O}{S}{B_i})$ the cut between the reachability ball of $S$ at time $B_i$ and the rest of the vertices. Firstly, if $B_i>2n\log{n}$ we immediately terminate the coupling and sample trials for $Y_i$ independently at random.
        Otherwise, we define $\text{CrossCount}(i)=\sum_{e \in \text{Cross}(i)}r-|\lambda^O_{[0,B_i]}|$ to be the number of labels left crossing the cut and reveal each snapshot of $\mathcal{G}^O$ in turn.
        On the $t$th subsequent round,
        \begin{itemize}
            \item If $t>2n \log{n}$, we declare the coupling to have failed.
            \item Otherwise, if the edge revealed from $\mathcal{G}^O$ crosses the cut we sample a success for $Y_i$.
            \item If the edge revealed from $\mathcal{G}^O$ does not cross the cut, we sample a success for $Y_i$ with probability
            $\frac{p_i-\frac{\text{2CrossCount(i)}}{rn(n-1)-2t}}{1-\frac{2\text{CrossCount(i)}}{rn(n-1)-2t}}$, and otherwise sample a failure. (Note that the success probability is well defined as long as both the numerator and denominator are positive. For $n$ sufficiently large as long as the coupling does not fail, this will always hold.)
        \end{itemize}
        By construction we have that, either, the coupling fails or the number of trials before a success for $Y_i$ is less than or equal to the number of edges revealed until one crossing the cut $(B_{i+1}-B_i)$ is. Therefore, either we have that $Y_i<B_{i+1}-B_i$ or, alternatively, $B_{i+1}>2n\log{n}$ (from \cref{lem: order double upper} probability $o(1)$). Thus, we find that $Y_i<B_{i+1}-B_i$ with high probability via a union bound over the failure of the coupling.\\
    
        We will now consider the mass coupling formed by sampling each $Y_i$ based on the same sample of $\mathcal{G}$, and show that with sufficiently high probability $Y_{i}<B_{i+1}-B_i$ for all $i \in \{0,...,\frac{n-2}{2}\}$.
        Immediately, by construction we find that each $Y$ random variable is independent of the rest but there is a very high degree of dependence between the failure condition for each individual coupling.
        In particular, for all $i \in \{0,...,\frac{n-2}{2}\}$, we have that  $2n\log{n}<B_{i+1}$ implies that $B_{\frac{n}{2}}>2n\log{n}$.
        Thus, applying \cref{lem: growth upper}, we have that the probability of ANY of the couplings failing is at most $n^{\frac{\log{2}-1}{5}}$.
        We therefore have the claim.
        \end{proof}
        It remains to bound $\sum_{i=0}^{\frac{n-2}{2}}Y_{i}$, which we do in the following claim.
        \begin{claim}
            \label{clm: Lower Geo}
            For any $0<\delta<1$,
            \[\Prob\left[\mathbf{Y}\leq \frac{(1-\beta)(1-\delta)n\log{n}}{2(1+\beta)}\right]=1-O(n^{2(1-\beta)(\delta+\log{(1-\delta)})}).\]
        \end{claim}
        \begin{proof}
        For this we will return to our old friend \cref{lem: Janson}, for which we require a bound on the expectation,
        \[\Exp[\mathbf{Y}]=\frac{n(n-1)}{2(1+\beta)}\sum_{i=\frac{n}{2}}^{n-1} \frac{1}{i(n-i)}=\frac{n-1}{2(1+\beta)}\sum_{i=\frac{n}{2}}^{n-1} \frac{1}{i}+\frac{1}{n-i}=\frac{n-1}{2(1+\beta)} \left(H_{n-1}+\frac{2}{n}\right).\]
        For any $\beta>0$ sufficiently small, for all $n$ large enough, by \cref{lem: Harmonics},
        \[\geq \frac{(1-\beta)n\log{n}}{2(1+\beta)}.\]
        Now for any $0<\delta<1$, we have that,
        \[\Prob\left[\mathbf{Y}<\frac{(1-\beta)(1-\delta)n\log{n}}{2(1+\beta)}\right]\leq \Prob[\mathbf{Y}\leq (1-\delta)\Exp[\mathbf{Y}]].\]
        Applying \cref{lem: Janson},
        \[\leq \exp\left[\left(\frac{2(1+\beta)}{n}\right)\left(\frac{(1-\beta)n\log{n}}{2(1+\beta)}\right)(\delta+\log{(1-\delta)})\right]\]
        \[=n^{2(1-\beta)(\delta+\log{(1-\delta)})}.\]
        \end{proof}
        Combining \cref{clm: lower coupling} and \cref{clm: Lower Geo}, we have the result via a union bound over the failure probabilities.
    \end{proof}

\bibliography{References-icalp}

\end{document}